\documentclass{lmcs}
\pdfoutput=1

\usepackage{lastpage}
\lmcsdoi{16}{2}{4}
\lmcsheading{}{\pageref{LastPage}}{}{}%
{Jun.~25,~2019}{May~13,~2020}{}


\usepackage[T1]{fontenc}
\usepackage{amsthm}
\usepackage[inline]{enumitem}
\usepackage{amsmath}
\usepackage{amssymb}
\usepackage{fancyvrb}
\usepackage[draft]{fixme}
\fxsetup{targetlayout=color}
\usepackage[all]{xy}
\usepackage[algoruled,vlined,english]{algorithm2e}
\usepackage{wrapfig}
\usepackage{color}

\newcommand{\tool}[1]{\textsc{#1}\xspace}  

\newcommand{\eps}{{\varepsilon}}
\newcommand{\pgt}[1]{{\tt #1}}
\newcommand{\lsem}{\mbox{$\lbrack\hspace{-0.3ex}\lbrack$}}
\newcommand{\rsem}{\mbox{$\rbrack\hspace{-0.3ex}\rbrack$}}
\newcommand{\sempar}[1]{\mbox{\lsem\pgt{#1}\rsem}}
\newcommand{\pass}{\texttt{:=}\,}
\newcommand{\semi}{\texttt{;}}

\newcommand{\X}{\texttt{X}}
\newcommand{\C}{\texttt{C}}

\newcommand{\Y}{\texttt{Y}}

\newcommand{\setS}{\mathcal{S}}
\newcommand{\setT}{\mathcal{T}}
\newcommand{\setB}{\mathcal{B}}
\newcommand{\setA}{\mathcal{A}}
\newcommand{\PB}{{\texttt{PB}}}

\newcommand{\join}{\sqcup}

\renewcommand{\vec}[1]{{\mathbf{#1}}}

\newcommand{\ppol}{{\texttt{\upshape MPol}}}
\newcommand{\abspol}{{\texttt{\upshape APol}}}
\newcommand{\tabspol}{{\texttt{\upshape $\tau$APol}}}
\newcommand{\absppol}{{\texttt{\upshape AMPol}}}
\newcommand{\tabsppol}{{\texttt{\upshape $\tau$AMPol}}}
\newcommand{\tpol}{{\texttt{\upshape $\tau$Pol}}}
\newcommand{\tppol}{{\texttt{\upshape $\tau$MPol}}}
\newcommand{\idppol}{\mathit{Id}}
\newcommand{\acirc}{\mathop{\bullet}}   

\newcommand{\taucomp}{\mathop{\star}}
\newcommand{\tra}[1]{\ensuremath{\alpha!({#1})}}

\newcommand\lepoly[0]{\sqsubset}  
\newcommand\gepoly[0]{\sqsupset}  

\newcommand{\wgt}[1]{\ensuremath{{\Vert #1 \Vert}}}

\newcommand\closure[1]{\ensuremath{{Cl(#1)}}}

\newcommand\sd[1]{\ensuremath{{{\textsc{SD}}(#1)}}}
\newcommand\sdpart[1]{\ensuremath{{\lfloor #1 \rfloor}}}
\newcommand\sdm[1]{\ensuremath{{\lfloor\!\lfloor #1 \rfloor\!\rfloor}}}
\newcommand\sdc[1]{{\sdm{#1}}}

\newcommand\unimono[1]{\dot{#1}}
\newcommand\nunimono[1]{\ddot{#1}}
\newcommand\mnl[0]{{:}} 

\newcommand\procSDL{\textsc{Solve}\xspace}

\newcommand{\trans}[3]{{#2\xrightarrow{#1}#3}}
\newcommand{\wtrans}[4]{{#2\xrightarrow{#1|#4}#3}}
\newcommand{\tseqs}[1]{{#1}^*}
\newcommand{\wtseqs}[1]{\mathcal{W}{#1}^*}
\newcommand{\tseq}[3]{{#1\stackrel{#2}{\leadsto}#3}}

\newcommand{\concrete}[1]{{\widetilde{#1}}}  
\newcommand{\concat}[0]{}  

\newcommand\myforall[1]{\ensuremath{\forall #1\ .\ }}

\newcommand{\eqdef}{\stackrel{def}{=\;}}
\newcommand\tuple[1]{\ensuremath{\langle #1\rangle}}
\newcommand\fourtuple[4]{\ensuremath{\langle #1,\ #2,\ #3,\ #4\rangle}}

\newcommand{\ints}{\mathbb{Z}}
\newcommand{\nats}{\mathbb{N}}

\begin{document}

\title{Tight Polynomial Worst-Case Bounds \texorpdfstring{\\}{} for Loop Programs\rsuper{*}}

\titlecomment{\lsuper{*}A preliminary version of this paper appeared in FoSSaCS 2019, LNCS 11425:80--97}

\author[A.M. Ben-Amram]{Amir M. Ben-Amram\rsuper{a}}
\address{\lsuper{a}Qiryat Ono, Israel}
\email{amirben@mta.ac.il}

\author[G.W. Hamilton]{Geoff Hamilton\rsuper{b}}
\address{\lsuper{b}School of Computing, Dublin City University, Ireland}
\email{hamilton@computing.dcu.ie}

\begin{abstract}
In 2008, Ben-Amram, Jones and Kristiansen showed that
for a simple programming language---representing non-deterministic
imperative programs with bounded loops,
and arithmetics limited to addition and multiplication---it is possible
to decide precisely whether a program has certain growth-rate
properties, in particular whether
a computed value, or the program's running time, has a polynomial growth rate.

A natural and intriguing problem was to move from answering the decision problem to giving a quantitative result, namely, a tight polynomial upper bound.
This paper shows how to obtain \emph{asymptotically-tight}, \emph{multivariate}, \emph{disjunctive} polynomial bounds for this class of programs.
This is a complete solution: whenever a polynomial bound exists it will be found.

A pleasant surprise is that the algorithm is quite simple; but it relies on some subtle reasoning.
 An important ingredient in the proof is the \emph{forest factorization theorem},
a strong structural result on homomorphisms into a finite monoid.
\end{abstract}

\maketitle

\section{Introduction}

One of the most important properties we would like to know about programs is their \emph{resource usage}, i.e., the amount of resources
(such as time, memory and energy) required for their execution. This information is useful during development, when performance bugs
and security vulnerabilities exploiting performance issues can be avoided. It is also particularly relevant for mobile applications,
where resources are limited, and for cloud services, where resource usage is a major cost factor.

In the literature, a lot of different ``cost analysis'' problems (also called ``resource bound analysis'', etc.)
have been studied (e.g.~\cite{Wegbreit:75,Rosendahl89,ACE,Albert-et-al:TCS:2011,APROVE-JAR2017,CiaoPP-TPLP2018,CHS:pldi2015,SZV:jar2017});
several of them may be grouped under the following general definition.
The \emph{countable resource problem} asks about the maximum usage of a ``resource'' that accumulates during execution,
and which one can explicitly count, by instrumenting the program with an accumulator variable and instructions to increment it where
necessary. For example, we can estimate the \emph{execution time} of a program  by counting certain ``basic steps''.
 Another example is counting the number of visits to designated program locations. Realistic problems of this type include bounding
the number of calls to specific functions, perhaps to system services; the number of I/O operations;  number of accesses to memory, etc.
The consumption of resources such as \emph{energy} suits our problem formulation as long as such explicit bookkeeping is possible (we have to assume that the
increments, if not constant, are given by a monotone polynomial expression).

In this paper we solve the \emph{bound analysis problem} for a particular class of programs, defined in~\cite{BJK08}.
The bound analysis problem is to find symbolic bounds on the maximal possible value of an integer variable
at the end of the program,
in terms of some integer-valued variables that appear in the initial state of a computation.
Thus, a solution to this problem might be used for any of the resource-bound analyses above.
 In this work we focus on values that
grow polynomially (in the sense of being bounded by a polynomial), and our goal is to find polynomial bounds that are tight, in the sense
of being precise up to a constant factor.

The programs we study are expressed by the so-called \emph{core language}. It is
an imperative language, including bounded loops, non-deterministic branches and restricted arithmetic expressions;
the syntax is shown in Figure~\ref{fig-syntax}. The semantics is explained and motivated below, but is largely intuitive;
see also the illustrative example in Figure~\ref{fig:intro-example}.
In 2008, it was proved~\cite{BJK08} that for this language it is decidable whether a computed result is polynomially bounded or not.  This makes the language
an attractive target for work on the problem of computing tight bounds. However, for the past ten years there has been no improvement on~\cite{BJK08}.
We now present an algorithm to compute, for every program in the language, and every variable in the program which
has a polynomial upper bound (in terms of input values), a tight  polynomial bound  on its largest attainable value (informally, ``the worst-case value'')
as a function of the input values. The bound is guaranteed to be tight up to a multiplicative constant factor but constants are left implicit (for example a bound
quadratic in $n$ will always be represented as $n^2$).
The algorithm could
be extended to compute upper and lower bounds with explicit constant factors, but choosing to ignore coefficients simplifies the algorithm considerably.
In fact, we have striven for a simple, comprehensible algorithm, and we believe that the  algorithm we present is sufficiently simple that, beyond
being comprehensible, offers insight into the structure of computations in this model.

Our philosophy is that research on complete solutions to static analysis questions regarding weak languages is desirable for several reasons.
First, it is theoretically satisfying---it establishes a clear and definite result. The algorithm can be possibly employed later in more complex situations,
and we will at least have an assurance that it does its part; arguments about the value of relying on decidable problems in program analysis have recently been
given in~\cite{McMillanPadon18}, and~\cite{Kincaid:sas18} gives a methodology for incorporating them as parts in a bigger system (albeit for safety problems).
Knowing that a problem is solvable for a certain weak language gives us a point of reference for future research, and makes it meaningful to further discuss
questions of computational complexity.
As pointed out in~\cite{McMillanPadon18}, when an algorithm to prove a property has a completeness proof, it usually means that it is possible to furnish a justification
for a negative answer, which is of value to the user. In our case, if the algorithm returns a bound which is higher than what you wanted, you can obtain from it
an \emph{execution pattern} which shows how the result arose.
Finally, the quest for complete solutions drives research forward by setting challenges which invite new insights and ideas.

Next, we explain the definition of the language in more detail. We will comment
about the motivation for the definitions, in particular vis-{\`a}-vis the analysis of fuller programming languages.  The main argument is that  choices in
the definition of the language are driven by the idea of using it as a
\emph{conservative abstraction}.

\subsection{The core language}%
\label{sec:language}

\hspace*{5cm} \\
\begin{figure}[htb]
\[\renewcommand{\arraystretch}{1.3}
\begin{array}{rcl}
\verb+X+\in\mbox{Variable} &\;\; ::= \;\; & \X_1 \mid\X_2 \mid \X_3 \mid
 \ldots  \mid \X_n\\
\verb+E+\in\mbox{Expression} & ::= & \verb+X+ \mid \verb/E + E/ \mid
\verb+E * E+\\
\verb+C+\in\mbox{Command} & ::= & \verb+skip+ \mid \verb+X:=E+
                                \mid \verb+C+_1 \semi \verb+C+_2
                                \mid \texttt{loop E \{C\}}
                         \mid \texttt{choose}\;  \C_1  \; \texttt{or} \; \C_2
 \end{array} \renewcommand{\arraystretch}{1.0}\]
\caption{Syntax of the core language.}%
\label{fig-syntax}
\end{figure}

\paragraph{\bf Data.}
The only type of data in the core language is non-negative integers.\footnote{This could be modified to all integers, as explained later.}
In a practical setting, a program may include  statements that manipulate
non-integer data that can, however, be abstracted away without losing the information critical to loop control
---hence to a complexity analysis---as loops are often controlled by integer variables.
In other cases, it is possible to preprocess a program to replace complex data values with their size (or ``norm''), which is the quantity of importance
for loop control. Methods for this process have been widely studied in conjunction with termination and cost analysis.
These considerations motivate the study of weak languages that handle integers.

\paragraph{\bf Command semantics.}
The core language is inherently non-deterministic.
The {\tt choose} command represents a non-deterministic choice, and can be used to abstract any concrete conditional command by simply ignoring the condition.
 Note that what we ignore is branches within a loop body
and not branches that implement the loop control, as loops are represented by a dedicated loop command.
The command $\verb+loop E {C}+$ repeats \pgt{C} a non-deterministic number of times
bounded by the value of $\pgt{E}$, which is evaluated just before the loop is entered. Thus, as a conservative abstraction, it
may be used to model different forms of loops (for-loops, while-loops)
as long as a bound on the number of iterations, as a function of the program state on loop initiation,
 can be determined and expressed in the language.
There is an ample body of research on analysing programs to find such bounds where they are not explicitly given by the programmer;
in particular,  bounds can be obtained from a \emph{ranking function} for the loop~\cite{PR:04,BagnaraHZ08,ADFG:2010,Ben-AmramG13jv,Ben-AmramG17}. 
Note that the arithmetic in our language is too restricted to allow for the maintenance of counters and the management of \emph{while} loops,
as there is no subtraction, no explicit constants and no tests. Thus, for realistic ``concrete'' programs which use such devices, loop-bound analysis is supposed to be
performed \emph{on the concrete program} as part of the process of abstracting it to the core language. This process is illustrated in~\cite[Sect.~2]{BAPineles:2016}.
The semantics of the loop is non-deterministic, so that the loop is allowed to
actually perform fewer iterations than indicated by the bound expression; this is useful both for modeling loops that can ``break,''
as well as for using the results of auxiliary analyses, as those usually provide just a bound, not a precise number of iterations.

An interesting observation has been made by Jones and Kristiansen~\cite{JK08}: typically, algorithms whose goal is to prove loop termination
and establish loop bounds do so by
focusing on values that \emph{decrease} (examples are the Size-Change Termination principle~\cite{leejonesbenamram01} and numerous
methods that discover \emph{ranking functions}).   In contrast, by abstracting to our core language (\cite{JK08} uses a very similar
language), we focus on values that \emph{grow}.  Recent work in static analysis~\cite{Giesl:toplas2016} describes an analysis system which
combines a subsystem for
loop-bound analysis (via ranking functions) with a subsystem for
growth-rate analysis, which establishes symbolic bounds on data that grow along a loop.
 Our definition of the core language separates the concerns and concentrates on the problem of value growth,
for a program (or program fragment) where loop bounds are already known.  Note however that the loop bound is to be given as a function
of the state in which the loop is started, and may well depend on values that are the result of previous computations, as our example program in Figure~\ref{fig:intro-example} illustrates.

\begin{figure}[htb]
\begin{Verbatim}[codes={\catcode`$=3\catcode`_=8}]
loop X$_1$ {
   loop X$_2$ + X$_3$ { choose { X$_3$:= X$_1$; X$_2$:= X$_4$ } or { X$_3$:= X$_4$; X$_2$:= X$_1$ } };
   X$_4$:= X$_2$ + X$_3$
};
loop X$_4$ { choose  { X$_3$:= X$_1$ + X$_2$ + X$_3$ } or { X$_3$:= X$_2$;  X$_2$:= X$_1$ } }
\end{Verbatim}
\caption{A core-language program. \texttt{loop} $n$ \pgt{C} means ``do \pgt{C} at most $n$ times.''}%
\label{fig:intro-example}
\end{figure}

From a computability viewpoint,
the use of bounded loops restricts the programs that can be represented to such that compute
primitive recursive functions; this is a rich enough class to cover a lot of useful algorithms and make
the analysis problem challenging. In fact, our language resembles a weak version of Meyer and Ritchie's LOOP language~\cite{MR:67},
which computes all the primitive recursive functions, and where behavioral questions like ``is the result linearly bounded'' are undecidable.

\subsection{The algorithm}

Consider the program in Figure~\ref{fig:intro-example}. Suppose that it is started with the values of the variables $\X_1,\X_2,\dots$~being $x_1,x_2,\dots$.
Our purpose is to bound the values of all variables at the conclusion of the program in terms of those initial values. Indeed, they are all polynomially bounded, and
our algorithm provides tight bounds. For instance, it establishes that the final value of $\X_3$ is tightly bounded (up to a constant factor) by
$\max (x_4(x_4 + x_1^2),x_4(x_2 + x_3 + x_1^2) )$.

Actually, the algorithm produces information in a more precise form, as \emph{a disjunction of simultaneous bounds}.
This means that it generates a set of
tuples of polynomials, called \emph{multi-polynomials}. Each such tuple provides simultaneous bounds on all variables in a subset of possible executions; for example,
with the program in Figure~\ref{fig:intro-example}, one such multi-polynomial is $\tuple{x_1, x_2, x_3, x_2+x_3}$, which means that, starting with the valuation
$\X_1\mapsto x_1,\ \X_2\mapsto x_2,\ \X_3\mapsto x_3, \X_4\mapsto x_4$ the value of $\X_4$ at the end
of some possible execution (specifically, one in which the second and the third loop commands both exit immediately) is tightly (in this case, exactly)
described by $x_2+x_3$; while the other three variable retain their initial values.
 Other multi-polynomials will represent other sets of possible executions, their union covering all executions.
This \emph{disjunctive} form is  important in the context of a compositional analysis. To see why, suppose that we provide, for a command with variables
$\X,\Y$, the bounds $\tuple{x,y}$ and $\tuple{y,x}$. Then we know that the \emph{sum} of their values is always bounded by $x+y$, a result that
would have not been deduced had we given the bound $\max(x,y)$ on each of the variables. The difference may be critical
for the success of analyzing an enclosing or subsequent command.

 \emph{Multivariate} bounds are often of interest, and perhaps require no justification,
 but let us point out that multivariate polynomials are necessary even if we are
ultimately interested in a univariate bound, in terms of some single initial value, say $n$.
 This is, again, due to the analysis being compositional.   When we analyze an internal command that uses variables
 $\X,\Y,\dots$ we do not know in what possible contexts the command will be executed and how the values of these variables will be related to $n$.

Some highlights of our solution are as follows.
\begin{itemize}
\item We reduce the problem of analyzing any core-language program to the problem of analyzing a single loop, whose body is
already processed, and therefore presented as a collection of abstract state-transitions. This is typical of algorithms that analyze a structured imperative
language and do so compositionally.
\item Since we are computing bounds only up to a constant factor, we work with \emph{abstract polynomials}, that have no numeric coefficients.
\item We further introduce \emph{$\tau$-polynomials}, to describe the evolution of values in a loop.
These have an additional
parameter $\tau$ (for ``time''; more precisely, number of iterations).
Introducing $\tau$-polynomials was a key step in the solution.
\item The analysis of a loop is simply a closure computation under two operations: ordinary composition, and \emph{generalization} which is the
operation that predicts the evolution of values by judiciously adding $\tau$'s to \emph{idempotent} abstract transitions.
\end{itemize}

\noindent
The remainder of this paper is structured as follows. In Section 2 we give some definitions and state our main result.
In Sections~\ref{sec:algorithm-prelim}---\ref{sec:closure-algorithm} we present our algorithm.
In Section~\ref{sec:correctness-sdl}, we give the correctness statement for our algorithm, and in Section~\ref{sec:correctness-main} we give the correctness proofs for this.
In Section~\ref{sec-complexity} we consider the computational complexity of our algorithm, and in Section~\ref{sec:extensions} we consider extensions to our algorithm and
open problems. Section~\ref{sec:rw} describes related work, and Section~\ref{sec:conclusion} concludes and discusses ideas for further work.

\section{Preliminaries}%
\label{sec:prelim}

In this section, we give some basic definitions, complete the presentation of our programming language and precisely state the main result.

\subsection{Some notation and terminology}

\paragraph{\bf The language} We remark that in our language syntax
there is no special form for a ``program unit;'' in the text we sometimes use ``program'' for the subject of our analysis,
yet syntactically it is just a command.

\paragraph{\bf Polynomials and multi-polynomials}
 We work throughout this article with
multivariate polynomials in $x_1,\dots,x_n$
that have non-negative integer coefficients and no variables other than $x_1,\dots,x_n$; when we speak of a polynomial we always mean one of this kind.
Note that over the non-negative integers, such polynomials
are monotonically (weakly) increasing in all variables.
Our algorithm sometimes deals with monomials, and the reader may assume that a polynomial is always represented as a non-redundant set of monomials
(i.e., $p(x)=2x$ is never represented as $x+x$ or $2x+0x^2$).

The post-fix substitution operator $[a/b]$ may be applied to any sort of expression containing a variable $b$, to substitute $a$ instead; e.g.,
$(x^2+yx+y)[2z/y] = x^2+2zx+2z$.

When discussing a command, state-transition, or program trace, with a variable $\X_{i}$, $x_i$ will denote, as a rule, the initial value of this variable,
and $x'_i$ its final value. Thus we distinguish the syntactic entity by the typewriter font.

The parameter $n$ always refers to the number of variables in the subject program. The set $[n]$ is
$\{1,\dots,n\}$.
For a set $S$ an $n$-tuple over $S$ is a mapping from $[n]$ to $S$.
The set of these tuples is denoted by $S^n$.
Throughout the paper, various natural liftings of operators to collections of objects are used, e.g., if $S$ is a set of integers then
$S+1$ is the set $\{s+1 \mid s\in S\}$ and $S+S$ is $\{s+t \mid s,t\in S\}$.  We use such lifting with sets as well as with tuples.
If $S$ is ordered, we extend the ordering to $S^n$ by comparing tuples element-wise (this leads to a partial order, in general, e.g., with natural numbers,
$\tuple{1,3}$ and $\tuple{2,2}$ are incomparable).

\begin{defi}
A function of the form $\tuple{ {\vec p}[1], \dots, {\vec p}[n] }$, i.e., an $n$-tuple of polynomials,
is called \emph{a multi-polynomial (MP)}.
We denote by $\ppol$ the set of multi-polynomials, namely ${(\nats[\vec x])}^n$,where the number of variables $n$ is fixed by context.
\end{defi}

\begin{defi}
A \emph{polynomial transition (PT)} is a computation that transforms an initial state $\vec x = \tuple{x_1,\dots,x_n}$ to a new state
${\vec x}' = \tuple{x'_1,\dots,x'_n} = {\vec p}({\vec x})$ where $\vec p \in \ppol$.
\end{defi}

The distinction between a MP and a PT is perhaps a bit philosophical: An MP is a function: a mathematical object that exists independently
of computation. A PT is a computation whose effect is described by such a function. We can say, for example, that the command
$\X_2 \verb/:= X/_2+\X_1$\ performs a PT\@.  Its associated MP is $\tuple{x_1, x_2+x_1}$.

Various operations will be applied to
MPs, mostly obvious---in particular, composition (which corresponds to sequential application of the transitions).
Note that composition of multi-polynomials, ${\vec q} \circ {\vec p}$, is naturally defined since $\vec p$ supplies $n$ values for the $n$
variables of $\vec q$, 
and we have:  $({\vec p}\circ{\vec q})[i]= {\vec p}[i] \circ {\vec q}$.
We define $\idppol$ to be the identity transformation, ${\vec x}' = {\vec x}$ (in MP notation: ${\vec p}[i] = x_i$ for $i=1,\dots,n$).

\subsection{Formal semantics of the core language}%
\label{sec:semantics}

The semantics
associates with every command
\verb+C+ over variables $\X_1,\dots,\X_n$
a relation $\sempar{C} \subseteq \nats^n \times \nats^n$.
 In the expression $\vec x \sempar{C} \vec y$,
vector $\vec x$ (respectively $\vec y$) is the state before (after) the
execution of \verb+C+.

The semantics of {\tt skip} is the identity. The semantics of an assignment
$\X_i\verb+:=E+$  associates to each state $\vec x$ a new state $\vec y$ obtained by replacing
the component $x_i$ by the value of the expression \verb+E+ when evaluated over state $\vec x$. This is defined in the
natural way (details omitted), and is denoted by $\sempar{E}{\vec x}$.
Composite commands are described by the straightforward  equations:
\begin{eqnarray*}
\lsem{\tt C}_1 {\tt ;C}_2\rsem &=&
  \sempar{C$_2$}\circ\sempar{C$_1$} \\
\lsem\verb+choose C+_1\verb+ or C+_2\rsem &=&
  \sempar{C$_1$}\cup\sempar{C$_2$} \\
\lsem\verb+loop E {C}+\rsem &=&
\{ (\vec{x},\vec{y}) \mid \exists i \le \sempar{E}{\vec x} :
\vec{x} \sempar{C}^i \vec{y} \}
\end{eqnarray*}
where $\sempar{C}^i$ represents $\sempar{C}\circ\cdots\circ\sempar{C}$
(with $i$ occurrences of $\sempar{C}$); and $\sempar{C}^0 = \idppol$.

\paragraph{\bf Remarks} The following two changes may enhance the applicability
of the core language for simulating  concrete programs; we include
them as ``options'' because they do not affect the validity of our proofs.
\begin{enumerate}
\item
The semantics of an assignment operation may be non-deterministic:
\verb+X := E+ assigns to {\tt X} some non-negative value \emph{bounded} by {\tt E}.
This is useful to abstract expressions which are not in the core language,
and also to use the results of size analysis of subprograms. Such an analysis
may determine invariants such as ``the value of \verb+f(X,Y)+ is at most the sum
of {\tt X} and {\tt Y}.''
The reason that this change does not affect our results is that we work
exclusively with monotone increasing functions.
This assignment semantics is called \emph{a lossy assignment} in~\cite{BK11}, due to an analogy
with lossy counter machines~\cite{Mayr03}.

\item\label{rmrk-ints}
The domain of the integer variables may be extended to $\ints$. In this case
the bounds that we seek are on the absolute value of the output in terms of
absolute values of the inputs. This change does not affect our conclusions
because of the facts $|xy| = |x|\cdot |y|$ and $|x+y| \le |x|+|y|$.
The semantics of the loop command may be defined either as doing nothing if the loop bound is not positive, or using the absolute value as a bound.

%
\end{enumerate}

\subsection{Detailed statement of the main result}%
\label{sec:goals}

The \emph{polynomial-bound analysis problem} is to find,
for any given command, which output variables
are bounded by a polynomial in the input values (which are simply the values of all variables
upon commencement of the program), and to bound these output values tightly (up to constant factors).
The problem of \emph{identifying} the polynomially-bounded variables is completely solved by~\cite{BJK08}.
We rely on that algorithm, which is polynomial-time, for doing this classification (as further explained in Section~\ref{sec:basic}).
Our main result is thus stated as follows.

 \begin{thm}%
 \label{thm:mainresult}
 There is an algorithm which, for a command \verb+C+, over variables $\X_1$ through $\X_n$,
 outputs a set $\setB$ of multi-polynomials, such that the following hold, where $\PB$ is the set of indices $i$ of variables
 $\X_i$ whose final value under $\sempar{C}$  is polynomially bounded.
\begin{enumerate}
\item (Bounding) There is a constant $c_{\vec p}$ associated with each ${\vec p}\in \setB$, such that
\begin{equation*}
 \myforall{\vec x, \vec y} \vec x \sempar{C} \vec y \Longrightarrow \exists {\vec p}\in\setB \,.\forall i\in\PB \,.\, y_i \le c_{\vec p}{\vec p}[i](\vec x)
 \end{equation*}
\item (Tightness) For every ${\vec p}\in \setB$ there are constants $d_{\vec p}>0$, ${\vec x}_0$ such that
for all ${\vec x}\ge {\vec x}_0$ there is a $\vec y$ such that
\[
  \vec x \sempar{C} \vec y \text{ and } \forall i\in\PB \,.\, y_i \ge d_{\vec p}{\vec p}[i](\vec x)
\]
\end{enumerate}
\end{thm}

\noindent
A remark regarding the decision to provide bounds only up to constant factors: based on the proof,
the algorithm could be extended to explicitly generate the coefficients $c_{\vec p}$ and $d_{\vec p}$ in the above statement; but ignoring
constant factors simplifies the algorithm significantly.  In essence, all the considerations regarding constants are removed from the algorithm to the proofs.
 Moreover, methodically, our goal was to show that
in the sense that we consider (i.e., ``big Theta'' bounds), the problem could be solved \emph{completely}.
We do not attempt to completely solve the problem of computing bounds with the utmost precise constants.
We assume that readers with a Computer Science background are aware of the ubiquity of ``big Theta'' bounds in the analysis of algorithms and dispense
with a discussion of justifications for using them.

\paragraph{\bf A comment on counters}

As discussed in the introduction, various \emph{countable resource problems}, such as bounding the number of program
steps, may be reduced to the bound analysis problem.
We'd like to point out that the reduction, by introducing a variable to count the events of interest, may be carried out  in our core
language, despite the lack of explicit constants. It suffices to reserve a dedicated variable, say \pgt{U}, to represent a unit of the resource.
Then, we advance the counter by adding \pgt{U}, and its worst-case bound will equal the worst-case bound of the resource, times \pgt{U}.
As an example, Figure~\ref{fig:instrumented} shows how the program from Figure~\ref{fig:intro-example}
could be instrumented for computing its time complexity. We use a fresh variable $\X_5$ as a ``unit'' and another one $\X_6$ to count the steps.
The algorithm will deduce that the final value of $\X_6$, $x_6'$, is tightly bounded by $(x_2+x_3+x_1^3 + x_1 x_4)x_5$.  Erasing the $x_5$ factor
we obtain a tight bound on the unit-cost time.
\begin{figure}[htb]
\begin{Verbatim}[codes={\catcode`$=3\catcode`_=8}]
loop X$_1$
   X$_6$ := X$_6$ + X$_5$;
   loop X$_2$ + X$_3$ { X$_6$ := X$_6$ + X$_5$;
                 choose { X$_3$:= X$_1$; X$_2$:= X$_4$ } or { X$_3$:= X$_4$; X$_2$:= X$_1$ }
                };
   X$_4$:= X$_2$ + X$_3$
};
loop X$_4$ { X$_6$ := X$_6$ + X$_5$;
         choose  { X$_3$:= X$_1$ + X$_2$ + X$_3$ } or { X$_3$:= X$_2$;  X$_2$:= X$_1$ }
        }
\end{Verbatim}
\caption{An instrumented core-language program.}%
\label{fig:instrumented}
\end{figure}

\section{Analysis Algorithm: First Concepts}%
\label{sec:algorithm-prelim}
\label{sec:tau-poly}  

The following sections describe our analysis algorithm.  Naturally, the most intricate part of the analysis concerns loops.
In fact we break the description into stages: first we reduce the problem of analyzing any program to that of analyzing
\emph{simple disjunctive loops}, defined next. Then, we approach the analysis of such loops, which is the main effort in this work.

\begin{defi}%
\label{def:sdl}
A \emph{simple disjunctive loop (SDL)} is  a finite set $\setS$ of polynomial transitions.
\end{defi}
The loop is ``disjunctive'' because its meaning is that in every iteration, any of the given transitions may be chosen.
A SDL does not specify the number of iterations; our analysis generates polynomials that depend on the number of iterations as well as the initial state. For this purpose, we now
introduce $\tau$-polynomials where $\tau$ represents the number of iterations.
\begin{defi}
$\tau$-polynomials are polynomials in $x_1,\dots,x_n$ and $\tau$.
\end{defi}
If $p$ is a $\tau$-polynomial, then $p(v_1,\dots,v_n)$ is the result of substituting each $v_i$ for the respective $x_i$;
and we also write $p(v_1,\dots,v_n,t)$ for the result of substituting $t$ for $\tau$ as well.
The set of $\tau$-polynomials in $n$ variables ($n$ known from context) is denoted $\tpol$.

Multi-polynomials and polynomial transitions are formed from $\tau$-polynomials just as previously defined and are used to represent the
effect of a variable number of iterations.
For example, the $\tau$-polynomial transition $\tuple{x'_1,x'_2} = \tuple{x_1,\ x_2+\tau x_1}$ represents the effect of repeating ($\tau$ times) the assignment
$\X_2 \verb/:= X/_2+\X_1$.   The effect of iterating the composite command
$\X_2 \verb/:= X/_2+\X_1\verb/; X/_3\verb/:= X/_3+\X_2$ is described by  ${\vec x}' = \tuple{x_1,\ x_2+\tau x_1, \ x_3 + \tau x_2 + \tau^2 x_1}$
(here we already have an upper bound which is not reached precisely, but is correct up to a constant factor).

We denote the set of $\tau$-polynomial transitions by $\tppol$.
Note that $\tau$ has a special status: as it does not represent a program variable, there is no component in the multi-polynomial for giving its value.
We should note that composition ${\vec q}\circ {\vec p}$ over $\tppol$ is performed by substituting ${\vec p}[i]$ for each occurrence of $x_i$ in $\vec q$.
Occurrences of $\tau$ are unaffected (since $\tau$ is not part of the state).

\section{Reduction to Simple Disjunctive Loops}%
\label{sec:basic}

We show how to reduce the problem of analysing core-language programs to the analysis of simple disjunctive loops.

\subsection{Symbolic evaluation of straight-line code}

Straight-line code consists of atomic commands---namely assignments (or \texttt{skip}, equivalent to
$\X_1\pass \X_1$), composed sequentially.  It is obvious that symbolic evaluation of such code leads to polynomial
transitions.

\begin{exa}
Consider the following command:
\begin{center}
$\X_2 \verb/:= X/_1\verb/; X/_4 \verb/:= X/_2\verb/ + X/_3\verb/; X/_1 \verb/:= X/_2\verb/ * X/_3$
\end{center}
This is precisely represented by the transition $\tuple{x_1,x_2,x_3,x_4}' = \tuple{ x_1x_3,\ x_1,\ x_3,\ x_1+x_3}$.
\end{exa}

\subsection{Evaluation of non-deterministic choice}

Evaluation of the command \[\texttt{choose}\; \C_1 \; \texttt{or} \; \C_2\] yields a set of possible
outcomes.
Hence, the result of analyzing a command will be a \emph{set} of multi-polynomial transitions. We express this in the common notation of
abstract semantics:
\[
{\sempar{C}}^{S} \in \wp(\ppol) \,.
\]
For uniformity, we consider $\sempar{C}^{S}$ for an atomic command to be a singleton in $\wp(\ppol)$.
Composition is naturally extended to sets, and the semantics of a choice command is now simply set union, so we have:
\begin{eqnarray*}
\lsem{\tt C}_1 {\tt ;C}_2\rsem^{S} &=&
  \sempar{C$_2$}^{S}\circ\sempar{C$_1$}^{S} \\
\lsem\verb+choose C+_1\verb+ or C+_2\rsem^{S} &=&
  \sempar{C$_1$}^{S}\cup\sempar{C$_2$}^{S}
\end{eqnarray*}

\begin{exa}
Consider the following command:
\begin{center}
$\X_2 \texttt{:= X}_1\texttt{; choose \{ X}_4 \texttt{:= X}_2 \texttt{ + X}_3 \texttt{ \} or \{ X}_1 \texttt{:= X}_2 \texttt{* X}_3 \texttt{ \}}$
\end{center}
This is represented by the set $\{ \tuple{ x_1, x_1, x_3, x_1+x_3}, \tuple{ x_1 x_3, x_1, x_3, x_4 } \}$.
\end{exa}

\subsection{Handling loops}%
\label{sec:solve}

The above shows that any loop-free command in our language can be \emph{precisely} represented by a finite set of MPs,
that is, its semantics $\sempar{C}$ equals the union of a finite number of PTs, which we can compute by symbolic evaluation.
Consequently, the problem of analyzing any loop command is reduced to the analysis of simple disjunctive loops, by first analyzing the loop body.
The goal of the loop analysis is to
produce a set of multi-polynomials that provide
tight worst-case bounds on the results of any number of iterations of the loop body---``tight'' in the sense of Theorem~\ref{thm:mainresult}
(see also Example~\ref{ex:solve} below).
  We will refer to this, concisely, as \emph{solving} the loop.

  Suppose that we have an algorithm \procSDL that takes a simple disjunctive loop and computes tight bounds for its results, with explicit dependence
  on the number of iterations using the parameter $\tau$.
We use it to complete the analysis of any program by the following definition:
\begin{eqnarray*}
\lsem\verb+loop E {C}+\rsem^{S} &=&
(\procSDL(\sempar{C}^{S}))[\sempar{E}/\tau] \ .
\end{eqnarray*}

Thus, the whole solution is constructed as an ordinary abstract interpretation, following the semantics of the language, except for what happens
inside \procSDL, which is the subtle part of the algorithm; this procedure accepts a set of multi-polynomials, describing the body of the loop,
and generates a set of $\tau$-MPs that finitely describes the result of all execution sequences of the loop, thanks to a \emph{generalization}
procedure that captures the dependence of results on the iteration count by means of the parameter $\tau$.  After completing this procedure, we replace
$\tau$ with the expression $\sempar{E}$, which gives the actual loop bound as a polynomial in $x_1,\dots,x_n$.

\begin{exa}%
\label{ex:solve}
Consider the following command:
\begin{center}
$\X_4 \texttt{:= X}_1\texttt{; loop X}_4 \texttt{ \{ X}_2 \texttt{:= X}_1\texttt{ + X}_2\texttt{; X}_3 \texttt{:= X}_2 \texttt{\}}$
\end{center}
Solving just the loop yields the set $\mathcal L = \{ \tuple{ x_1, x_2, x_3, x_4}$, $\tuple{ x_1, x_2+\tau x_1, x_2+\tau x_1, x_4 } \}$
(the first MP accounts for zero iterations, the second covers any positive number of iterations).
We can now compute the effect of the given command as:
\begin{align*}
 {\mathcal L}[x_4/\tau] \circ \sempar{X$_4$  := X$_1$}^{S} &=
 {\mathcal L}[x_4/\tau] \circ \{ \tuple{ x_1, x_2, x_3, x_1} \}  \\ &=
   \{ \tuple{ x_1, x_2, x_3, x_1},  \tuple{ x_1, x_2+x_1^2, x_2+x_1^2, x_1 } \}.
\end{align*}
We can now demonstrate that the analysis of a loop may involve approximation.  Suppose that we slightly modify the code as follows:
\begin{center}
$\X_4 \texttt{:= X}_1\texttt{; loop X}_4 \texttt{ \{ X}_3 \texttt{:= X}_2\texttt{; X}_2 \texttt{:= X}_1\texttt{ + X}_2 \texttt{\}}$
\end{center}

Our algorithm will produce the same MPs, although a careful examination shows that the precise result of the loop (when it executes $x_1>0$ iterations)
is now given by
$\tuple{ x_1, x_2+x_1^2, x_2+x_1(x_1-1), x_1 }$.  This expression is not a MP (since we only consider polynomials with non-negative coefficients),
but is tightly approximated by the MP we produced, as $x_1(x_1-1) = \Theta(x_1^2)$.  The precise statement of what it means to tightly approximate
all possible results of a program by a finite set of MPs is the essence of Theorem~\ref{thm:mainresult}.
\end{exa}

The next section describes the procedure \procSDL, and operates under the assumption that all variables are polynomially bounded in the loop.
However, a loop can generate exponential growth. To cover this eventuality, we first apply the algorithm of~\cite{BJK08} that identifies which
variables are polynomially bounded. If some $\X_i$ is \emph{not} polynomially bounded we replace the $i$th component of all the loop transitions
with $x_n$ (we assume $x_n$ to be a dedicated, unmodified variable). Clearly, after this change, all variables are polynomially bounded;
moreover, variables which are genuinely polynomial are unaffected, because they cannot depend on a super-polynomial quantity (given the restricted
arithmetics in our language).
In reporting the results of the algorithm, we should display ``\texttt{super-polynomial}'' instead of any expression that includes $x_n$.

\section{Simple Disjunctive Loop Analysis Algorithm}%
\label{sec:closure-algorithm}

Intuitively, evaluating  $\verb/loop E {C}/$ abstractly
consists of simulating any finite number of iterations, i.e., computing
\begin{equation}
\label{eq:closure-seq}
Q_i  = \{\idppol\} \cup P \cup (P\circ P ) \cup \dots \cup {P}^{(i)}
\end{equation}
where $P =  \sempar{C}^{S} \in \wp(\ppol)$.  The question now is whether the sequence~\eqref{eq:closure-seq} reaches a fixed point.
In fact, it often doesn't. However, it is quite easy to see that in the \emph{multiplicative
fragment} of the language, that is, where the addition operator is not used, such non-convergence is associated with exponential growth. Indeed, since there is no
addition, all our polynomials are monomials with a leading coefficient of 1 (\emph{monic monomials})---this is easy
to verify. It follows that if the sequence~\eqref{eq:closure-seq} does not converge, higher and higher exponents must
appear, which indicates that some variable cannot be bounded polynomially. Taking the contrapositive, we conclude
that if all variables are known to be polynomially bounded the sequence will converge.  Thus we have the following easy (and not so satisfying) result:

\begin{obs}
For a SDL $P$ that does not use addition, the sequence $Q_i$ as in~\eqref{eq:closure-seq} reaches a fixed point, and the fixed point provides
tight bounds for all the polynomially-bounded variables.
\end{obs}

\begin{exa}
The following loop is in the multiplicative fragment, and has no exponential behaviour:
\[ \verb/loop X/_3\verb/ { X/_1 \verb/:= X/_2\verb/ * X/_2 \verb/;  X/_2 \verb/:= X/_3 \verb/ }/ \]
The effect of an iteration is given by the multi-polynomial ${\vec p} = \tuple{x_2^2,\ x_3,\ x_3}$.
Let $P=\{{\vec p}\}$.
The accumulating effect of the loop is given by
the union of:
\begin{align*}
\{&\idppol\} \\
P &= \{ \tuple{x_2^2,\ x_3,\ x_3} \} \\
P\circ P &= \{ \tuple{ x_3^2,\ x_3,\ x_3} \} \\
P^{(3)} &= P^{(2)} \\ &\hspace{-0.5em} \dots
\end{align*}
That is, the results are completely characterized by the three MPs above.
\end{exa}

\noindent
When we have addition, we find that knowing that all variables are polynomially bounded does not imply
convergence of the sequence~\eqref{eq:closure-seq}.  An example is
$\texttt{loop X}_3\texttt{ \{ X}_1 \texttt{:= X}_1\texttt{ + X}_2 \texttt{ \}}$
yielding the infinite sequence of MPs $\tuple{x_1,x_2,x_3}$, $\tuple{x_1+x_2,x_2,x_3}$, $\tuple{x_1+2x_2,x_2,x_3}$, \dots~%
Our solution employs two means. One is the introduction of $\tau$-polynomials, already presented.  The other
is a kind of \emph{abstraction}---intuitively, ignoring the concrete values of (non-zero) coefficients.
Let us first define this abstraction:

\begin{defi}\label{def:abspol}
$\abspol$, the set of abstract polynomials, consists of formal sums of distinct monomials over $x_1,\dots,x_n$,
where the coefficient of every monomial included is $1$.
We extend the definition to an abstraction of $\tau$-polynomials,
denoted $\tabspol$.
\end{defi}


\noindent
The meaning of abstract polynomials is given by the following rules:
\begin{enumerate}
\item\label{def:alpha}
The abstraction of a polynomial $p$, $\alpha(p)$, is obtained by modifying all (non-zero) coefficients to 1.
\item Addition and multiplication in $\tabspol$ are defined in a natural way so that $\alpha(p)+\alpha(q) = \alpha(p+q)$
and $\alpha(p)\cdot \alpha(q) = \alpha(p\cdot q)$ (to carry these operations out, you just go through the motions of adding or
multiplying ordinary polynomials, ignoring the coefficient values).
\item\label{def:gamma}
The \emph{canonical concretization} of an abstract polynomial, $\gamma({\vec p})$  is obtained by simply regarding it
as an ordinary polynomial.
\item These definitions extend to tuples of (abstract) polynomials in the natural way.
\item The set of abstract multi-polynomials $\absppol$ and their extension with $\tau$ ($\tabsppol$) are defined
as $n$-tuples over $\abspol$ (respectively, $\tabspol$). We use AMP as an abbreviation for abstract multi-polynomial.
\item Composition ${\vec p}\acirc {\vec q}$, for ${\vec p}, {\vec q} \in \absppol$ (or $\tabsppol$)
 is defined
as $\alpha(\gamma({\vec p})\circ \gamma({\vec q}))$; it is easy to see that one  can perform the calculation without the detour through polynomials
with coefficients.
The different operator symbol (``$\acirc$'' versus ``$\circ$'') helps in disambiguating expressions.
\end{enumerate}
An abstract polynomial can be reduced by deleting dominated monomials (e.g.,  $x^2+x \to x^2$); which clearly
preserves its meaning when viewed in ``big Theta'' terms. We expect that the combinatorial explosion in the algorithm can be attenuated
by reducing every abstract polynomial, but in this paper, we ignore this optimization to alleviate the description and proof of the algorithm.
Another useful optimization (which we also ignore in our analysis) is to delete dominated \emph{multi-polynomials} in a set of MPs.

\paragraph{\bf Analysing a SDL} To analyse a SDL specified by a set of MPs $\setS$, we start by computing $\alpha(\setS)$. The rest of the algorithm
computes within $\tabsppol$.
We  define two operations that are combined in the analysis of loops.
The first, which we call \emph{closure}, is simply the fixed point of accumulated iterations as in the multiplicative case. It is introduced by the following
two definitions.

\begin{defi}[iterated composition]%
\label{def:itcomp}
Let $\vec{t}$ be any abstract $\tau$-MP\@.  We define ${\vec t}^{\acirc (k)}$, for $k\ge 0$, by:
\begin{align*}
{\vec t}^{\acirc (0)} & = \idppol \\
{\vec t}^{\acirc (k+1)} & =  {\vec t}\acirc {\vec t}^{\acirc (k)} .
\end{align*}
For a set $\setT$ of  abstract $\tau$-MPs, we define, for $k\ge 0$:
\begin{align*}
\setT^{\acirc (0)} & = \{\idppol\} \\
\setT^{\acirc (k+1)} & =  \setT^{\acirc (k)}  \cup \    \bigcup_{ \vec q \in \setT,\ \vec p \in \setT^{\acirc (k)}} {\vec q}\acirc {\vec p} \, .
\end{align*}
\end{defi}
\noindent
Note that ${\vec t}^{\acirc (k)} = \alpha({\gamma(\vec{t})}^{(k)})$, where ${\vec p}^{(k)}$ is defined using ordinary composition.

\begin{exa}
Let $\vec p = \tuple{x_1, x_1+x_2}$; clearly its $k$th iterate (for $k\ge 1$) is ${\vec p}^{(k)} = \tuple{x_1, x_1+kx_2}$. However, with abstract polynomials,
we get ${\alpha({\vec p})}^{\acirc (k)} = \tuple{x_1, x_1+x_2}$; the growth in the coefficient is abstracted away.
Similarly, we have ${\{\alpha(\vec p)\}}^{\acirc (k)} = \{\idppol, \alpha({\vec p})\}$ for all $k\ge 1$.
\end{exa}

\begin{defi}[abstract closure]
 For finite
$P\subset \tabsppol$, we define:
\[
\closure{P} = \bigcup_{i=0}^{\infty} P^{\acirc (i)} \,.
\]
\end{defi}
In the correctness proof, we argue that when all variables are polynomially
bounded in a loop $\setS$, the closure of $\alpha(\setS)$ can be computed in finite time; equivalently, it equals
$\bigcup_{i=0}^{k} {(\alpha({\setS}))}^{\acirc (i)}$ for some $k$.
The argument for finiteness is essentially the same as in the multiplicative case.


The second operation is called \emph{generalization} and its role is to capture the behaviour of accumulator
variables, meaning variables that grow by accumulating increments in the loop, and make explicit the dependence on the
number of iterations.   The identification of which additive terms in a MP should be considered as increments that accumulate
is at the heart of our problem.

\begin{exa}%
\label{ex:motivate-tau}
Consider the following loop:
\begin{center}
$\texttt{loop \ldots\ \{ X}_1\texttt{:= X}_1\texttt{ + X}_3\texttt{; X}_2\texttt{:= X}_2\texttt{ + X}_3\texttt{ + X}_4\texttt{;  X}_4\texttt{:= X}_3\texttt{\}}$
\end{center}
We have omitted the loop bound in order to emphasize that we are now analyzing the \emph{growth of values as the loop progresses} rather
than its final result. The generalization construction below introduces $\tau$'s into the AMPs, obtained from the loop body, in order to express
this growth in terms of the number of iterations. Informally, consider the assignment to $\X_1$: it adds $x_3$ (the initial value of variable $\X_3$),
and does so on every iteration: so we have an arithmetic series with an increment of $x_3$. After $\tau$ iterations, $\X_1$  will have grown by $\tau x_3$.
The assignment to $\X_2$ similarly adds the values of two other variables, $\X_3$ and $\X_4$, but as $\X_4$ is rewritten with $x_3$ already at the end of the first
iteration, subsequent iterations add $2x_3$ (not $x_3+x_4$).  This will be expressed by introducing the term $\tau x_3$ (the coefficient 2 is abstracted away).
We return to this example, and the issue it illustrates, further below.
\end{exa}

The definition of ${\vec p}^\tau$ will is greatly simplified by concentrating on idempotent AMPs, defined next. We will later give
an example that shows that our definition does not work for AMPs which are not idempotent; while the fact that a complete solution can
be obtained by generalizing only idempotent elements comes out of a Ramsey-like property of finite monoids~\cite{Simon:TCS:90}, as will be
seen in Section~\ref{sec-proof-ub}.

\begin{defi}
$\vec p \in \tabsppol$ is called \emph{idempotent} if $\vec p \acirc \vec p = \vec p$.
\end{defi}
\noindent
Note that this is composition in the abstract domain. So, for instance, $\tuple{x_1, x_2}$ is idempotent, and so is
$\tuple{x_1 + x_2, x_2}$, while $\tuple{x_1x_2, x_2}$ and $\tuple{x_1 + x_2, x_1}$ are not.

\begin{defi}
For $\vec p$ an (abstract) multi-polynomial,
we say that $x_i$ is \emph{self-dependent} in $\vec p$ if $\vec p[i]$ depends on $x_i$.
We also say that the entry $\vec p[i]$ is self-dependent; the choice of term depends on context and the meaning should be clear either way.
We call a monomial self-dependent if all the variables appearing in it are.
We denote by $\sd{\vec p}$ the set $\{i \,:\, x_i \text{ is self-dependent in } \vec p\}$.
\end{defi}

We later show that in polynomially-bounded loops, if $x_i$ is self-dependent then $\vec p[i]$ must include the monomial $x_i$. To illustrate the
significance of self-dependent monomials, let us consider an
example where
${\vec p}[1] = x_1 +\dots$ and ${\vec p}[2] = x_2 + \dots$, and further ${\vec p}[3] = x_1 x_2 + \dots$. Then the monomial $x_1 x_2$ reappears in
every iterate ${\vec p}^{\acirc (k)}$. This is because all the variables in it are self-dependent.

\begin{defi}%
\label{def:MPnotation}
We define a notational convention for $\tau$-MPs, specifically for self-dependent entries of the MP\@.  Assuming that $x_i$ appears in ${\vec p}[i]$, we write:
\[ {\vec p}[i] = x_i + \tau {\vec p}[i]' + {\vec p}[i]''+ {\vec p}[i]''' \,,\] 
where ${\vec p}[i]'''$ includes all the non-self-dependent monomials of ${\vec p}[i]$, while the self-dependent monomials 
(other than $x_i$) are grouped into two sums:
$\tau{\vec p}[i]'$, including all monomials with a positive degree of $\tau$, and ${\vec p}[i]''$ which includes all the $\tau$-free monomials.
\end{defi}

\begin{exa}%
\label{ex:MPnotation}
Let $\vec p = \tuple{ x_1 + \tau x_2 + \tau x_3 +  x_3 x_4,\   x_3,\ x_3,\ x_4 }$. Then $\sd{\vec p} = \{1,3,4\}$.
Since $x_1$ is self-dependent, we will apply the above definition to ${\vec p}[1]$, so that
 ${\vec p}[1]' = x_3$, ${\vec p}[1]'' = x_3 x_4$ and ${\vec p}[1]''' = \tau x_2$. 
 Note that a factor of $\tau$ is stripped in ${\vec p}[1]'$.  Had the monomial been $\tau^2 x_3$, we would have ${\vec p}[1]' = \tau x_3$.
 \end{exa}

\begin{defi}[generalization]
Let $\vec p$ be idempotent in $\tabsppol$; define
${\vec p}^\tau$ by: \[
 {\vec p}^\tau [i] = \begin{cases}
   x_i + \tau {\vec p}[i]' + \tau{\vec p}[i]''+ {\vec p}[i]''' &  \text{if $i\in \sd{\vec p}$} \\ 
   {\vec p}[i] & \text{otherwise.}
\end{cases}
 \]

 \end{defi}
 \noindent
 Note that the arithmetic here is abstract (see examples below). Note also that in the term $\tau {\vec p}[i]'$ the $\tau$ is already present in $\vec p$,
while in $\tau {\vec p}[i]''$ it is added to existing monomials.  In this definition, the monomials of ${\vec p}[i]'''$ are treated like those of $\tau {\vec p}[i]'$; 
however, in certain steps of the proofs we will treat them differently, which is why the notation separates them. We next give a couple of examples, just
to illustrate the definition; we later (Page~\pageref{cog}) motivate the definition in the context of the loop analysis.

\begin{exa}%
\label{ex:tauisation:1}
Let $\vec p = \fourtuple{ x_1 + x_3}{ x_2 + x_3 + x_4}{ x_3}{ x_3 }$. Note that it corresponds to the loop body from Example~\ref{ex:motivate-tau}.
Further note that ${\vec p}\acirc {\vec p} = {\vec p}$, i.e., $\vec p$ is idempotent.
We have ${\vec p}^\tau = \fourtuple{ x_1 + \tau x_3}{ x_2 + \tau x_3+x_4}{ x_3}{ x_3 }$.
 \end{exa}

\begin{exa}\label{ex:tauisation:2}
Let $\vec p = \tuple{ x_1 + \tau x_2 + \tau x_3 + \tau x_3 x_4,\   x_3,\ x_3,\ x_4 }$.
 Note that ${\vec p}\acirc {\vec p} = {\vec p}$.
 The self-dependent variables are all but $x_2$.
 We have
${\vec p}^\tau = \tuple{ x_1 + \tau x_2 + \tau x_3 + \tau x_3 x_4,\   x_3,\ x_3,\ x_4 } = {\vec p}$.
 \end{exa}


\noindent
Finally we can present the analysis of the loop command.
\begin{algo} \procSDL{}$(\setS)$ \\
 Input: $\setS$, a polynomially-bounded disjunctive simple loop. \\
 Output: a set of $\tau$-MPs that tightly approximates the effect of $\tau$ iterations of loop $\setS$.
\end{algo}
\begin{enumerate}
\item Set $T = \alpha(\setS)$.

\item Repeat the following steps until $T$ remains fixed:
\begin{enumerate}
\item Closure: Set $T$ to $\closure{T}$.
\item Generalization: For all ${\vec p}\in T$ such that ${\vec p}\acirc {\vec p} = {\vec p}$, add ${\vec p}^\tau$ to $T$.
\end{enumerate}
\end{enumerate}

\begin{exa}
Consider the following loop:
\begin{center}
$\texttt{loop \ldots\ \{ X}_1\texttt{:= X}_1\texttt{ + X}_2\texttt{; X}_2\texttt{:= X}_2\texttt{ + X}_3\texttt{; X}_4\texttt{:= X}_3\texttt{ \}}$
\end{center}
The body of the loop is evaluated symbolically and yields the multi-polynomial:
\[
 \vec p = \fourtuple{ x_1+x_2}{ x_2+x_3}{ x_3}{ x_3 }
 \]
 Now, computing within $\absppol$,
 \begin{align*}
{\alpha({\vec p})}^{\acirc (2)} &=  \alpha({\vec p}) \acirc \alpha({\vec p}) = \fourtuple{ x_1+  x_2+x_3}{ x_2+ x_3}{ x_3}{ x_3 } ; \\
 {\alpha({\vec p})}^{\acirc (3)} &=  {\alpha({\vec p})}^{\acirc (2)}  \,.
 \end{align*}
 Here the closure computation stops.
 Since $\alpha( {\vec p}^{\acirc (2)} )$ is
 idempotent, we compute
\[
{\vec q} = {({\alpha({\vec p}) }^{\acirc (2)})}^\tau  = \fourtuple{ x_1+\tau x_2 + \tau x_3 }{ x_2+\tau x_3 }{ x_3}{ x_3 }
\]
and applying closure again, we obtain
\[\begin{array}{ll}
{\vec q}  \acirc {\alpha({\vec p}) } &= \fourtuple{ x_1+ x_2 + x_3 + \tau x_2 + \tau x_3 }{ x_2+ x_3 +\tau x_3 }{ x_3}{ x_3 } \\
{({\vec q})}^{\acirc (2)} &= \fourtuple{ x_1+ \tau x_2 + \tau x_3  + \tau^2 x_3 }{ x_2 +\tau x_3 }{ x_3}{ x_3 } \\
\end{array}
\]
where the first one simplifies to $\vec q$ by deleting dominated terms, and the second to
$\fourtuple{ x_1+ \tau x_2 + \tau^2 x_3 }{ x_2 +\tau x_3 }{ x_3}{ x_3 }$.
The last element is idempotent but applying generalization does not generate anything new. Thus the algorithm ends.
 The reader may reconsider the source code to verify that we have indeed obtained tight bounds for the loop.
 \end{exa}

 Note that when a program contains nested loops, innermost loops will be processed first and result in a set of abstract polynomials, so we might actually
 analyze any enclosing commands entirely in the abstract domain. This means rephrasing our definition of SDL by defining the input set as a set of AMPs rather
 than concrete polynomials; then the initial step above, where $\setS$ is abstracted, is skipped.  In the forthcoming sections we ignore this issue and
 continue to treat the input as being a set of concrete polynomials; this is for convenience  only. It is easy enough to restate the results such that the input
 is understood as abstract, but known to tightly describe the effect of some concrete piece of code.

\paragraph{\bf Comments on generalization}%
\label{cog}
The precise definition of the generalization operator has been one of the key steps in the development of this
algorithm (operators which are similar---but insufficient for our purpose---appear in  related work~\cite{KasaiAdachi:80,
KN04,NW06,JK08,BJK08,Giesl:toplas2016}). Let us attempt to give some intuition for its definition (not a correctness proof, yet).  Say $x_i$ is
self-dependent. An important point is the
partition of  the terms added to $x_i$ into $\tau {\vec p}[i]'$, ${\vec p}[i]''$ and ${\vec p}[i]'''$.  To see why 
the non-self-dependent monomials ${\vec p}[i]'''$ are not multiplied by $\tau$, consider Example~\ref{ex:tauisation:1}.  Had we been too eager to 
insert $\tau$'s and do this to non-self-dependent monomials as well, we would have got:
\[
  \fourtuple{ x_1 + \tau x_3}{ x_2 + \tau x_3 + \tau x_4}{ x_3}{ x_3 } \,
  \]
 (differing from the correct result in the second component). This would give a sound, but loose (overestimated) upper bound: as already discussed,
 when the transition represented by $\vec p$ is iterated, copies of the initial value $x_4$ do \emph{not} accumulate.  Next, the distinction of
 ${\vec p}[i]''$ from $\tau {\vec p}[i]'$ prevents $\tau$'s from being added where they already appear; this is important because we reapply generalization
 to results of previous steps while analyzing a given loop.
 Finally, an important aspect particular to our algorithm (compared to the above-cited) is the application of the generalization
 operator only to idempotent elements. Again, we can illustrate with an example that applying it too eagerly (to all AMPs in the closure) would be
 incorrect. Consider the following AMP:\@
 \[
 \vec p = \tuple{ x_1 +  x_2 + x_3 +  x_4,\   x_3,\ x_4,\ x_4 }
 \]
Generalization (incorrectly applied to $\vec p$, since it is not idempotent) would yield:
\[
  \fourtuple{ x_1 + x_2 + x_3 + \tau x_4}{ x_3}{ x_4}{ x_4 }
  \]
  This is bad because if $\vec p$ is iterated to accumulate copies of $x_4$ in $\X_1$ (a behaviour represented by the term $\tau x_4$), the value of $\X_2$ will no longer be $x_3$ (but $x_4$). Hence, this AMP
  describes a result that is not realizable (and when we substitute the loop bound for $\tau$, we will get a loose upper bound).

\paragraph{\bf A comment on linear bounds}
When values accumulate in a loop, a non-linear result is obtained. Stated contrapositively: linear results do not involve accumulation.  Indeed,
it is not hard to verify that if for linear polynomials (only) we maintain precise numeric coefficients, our algorithm would still converge.  Therefore, if we wish,
we can modify our algorithm so that \emph{for linear results we have bounds in which the coefficients are explicit and precise}.
This is not a strong result as might seem at first, since it is due to the weakness of our language: linear results can only be built up by a finite number of additions.
Still, if the algorithm is employed as a brick in a larger system, this property might possibly be useful.

\section{Correctness Statement for Simple Disjunctive Loops}%
\label{sec:correctness-sdl}

We claim that our algorithm obtains a description of the worst-case results of the program that is
precise up to constant factors.  That is, we claim that the set of MPs returned provides a ``big O'' upper bound (on all executions) which is also tight;
tightness means that every MP returned is also a lower bound (up to a constant factor) on an infinite sequence of possible executions.
 The main, non-trivial part of the algorithm is of course the solution of a simple disjunctive loop, procedure \procSDL. 
Completing this to show correctness for an arbitrary program is not difficult.

In this section we formulate the \emph{correctness statement for Simple Disjunctive Loops}, in other words, we state the contract
for procedure \procSDL, which, when satisfied, justifies its proper functioning within the whole algorithm.
This formulation forces us to step up the level of detail; specifically, we introduce \emph{traces}, to give a more concrete semantics to loops
(compared to Section~\ref{sec:semantics}).


\begin{defi}%
\label{def:trace}
Let $\setS$ be a set of polynomial transitions.
An \emph{(abstract) trace} over $\setS$ is a finite
sequence $\sigma = {\vec p}_1;\dots; {\vec p}_{|\sigma|}$ of elements of $\setS$.
Thus $|\sigma|$ denotes the \emph{length} of the trace.
The set of all traces is denoted $\tseqs{\mathcal S}$.
We write $\lsem{\sigma}\rsem$ for the composed relation
$ {\vec p}_{|\sigma|} \circ \dots\circ {\vec p}_1$ (for the empty trace, $\eps$, we have $\lsem{\eps}\rsem = \idppol$).
\end{defi}


\noindent
Using the following definition we will be able to give the desired correctness statement for \procSDL. 

\begin{defi}%
\label{def:unimono}
Let $p(\vec x)$ be a (concrete or abstract) $\tau$-polynomial.  We write $\unimono{p}$ for the \emph{linear monomials} of $p$, namely any one of
the form $ax_i$ for a constant coefficient $a$.  
We write $\nunimono{p}$ for the rest. Thus $p = \unimono{p}+\nunimono{p}$.
\end{defi}

\begin{thm}[Solution of disjunctive loop problem]%
\label{thm:sdlp}
Given a polynomially-bounded SDL
represented as a set $\setS$ of MPs, procedure \procSDL finds
a finite set $\setB$ of $\tau$-MPs which \emph{tightly bound} all traces over $\setS$. More precisely,
it guarantees:
\begin{enumerate}
\item (Bounding) There is a constant $c_{\vec p} > 0$ associated with each ${\vec p}\in \setB$, such that
\begin{equation*}
 \myforall{\vec x, \vec y, \sigma} \vec x \lsem\sigma\rsem \vec y \Longrightarrow \exists {\vec p}\in\setB \,.\, {\vec y} \le c_{\vec p}{\vec p}(\vec x, |\sigma|)
 \end{equation*}
\item (Tightness) For every ${\vec p}\in \setB$ there are constants $d_{\vec p}>0$, ${\vec x}_0$
such that for all ${\vec x}\ge {\vec x}_0$ there are a trace $\sigma$ and a state vector $\vec y$ such that
\[
 \vec x \lsem\sigma\rsem \vec y  \, \land \, {\vec y} \ge
           \unimono{\vec p}({\vec x}, |\sigma|)+ d_{\vec p}\nunimono{\vec p}({\vec x}, |\sigma|)
           \,.
\]
\end{enumerate}
\end{thm}
\noindent
Note that in the lower-bound clause (2), the linear monomials of $p$ are not multiplied by the constant $d_{\vec p}$; this
sets, in a sense, a stricter requirement for them: if the trace maps $x$ to $x^2$ then the bound
$2x^2$ is acceptable, but if it maps $x$ to $x$, the bound $2x$ is not accepted. The reader may understand this technicality by considering the
effect of iteration: it is important to distinguish the transition $x'_1 =  x_1$, which can be iterated ad libitum, from the transition
$x'_1 = 2x_1$, which produces exponential growth on iteration. Distinguishing $x'_1 = x_1^2$ from $x'_1 = 2x_1^2$ is not as important.
   We remark  that $c_{\vec p}, d_{\vec p}$ range over real numbers. However, our data and the coefficients of polynomials remain integers, it is only such
comparisons that are performed with real numbers (specifically, to allow $d_{\vec p}$ to be smaller than one; in the upper bounds, it is possible to stick
to integers).
Note also that polynomial boundedness is ensured by our algorithm
before applying the procedure (Section~\ref{sec:solve}), so the precondition of the correctness theorem is satisfied.

\section{Correctness Proofs}%
\label{sec:correctness-main}

 In this section we prove the correctness of the main, non-trivial part of our algorithm, namely the solution of a simple disjunctive loop,
 showing that it satisfies the requirements set forth in the last section.

 \paragraph{\textbf{Overview of the proof}}
 Intuitively, what we want to prove is that the multi-polynomials we compute cover all ``behaviors'' of the loop. More precisely, in the upper-bound part of the
proof we want to cover all behaviors: upper-bounding is a universal statement. To prove that bounds are tight, we show that each such bound constitutes a
\emph{lower bound} on a certain ``worst-case behavior'': tightness is an existential statement. The main aspects of these proofs are as follows:
\begin{itemize}
\item A key notion in our proofs is that of \emph{realizability}. Intuitively, when we come up with a bound, we want to show that there are traces that achieve (realize) this bound for arbitrarily large input values.
\item In the lower-bound proof, we describe a set of traces by a \emph{pattern}. A pattern is constructed like a regular expression
with concatenation and Kleene-star.
However, they allow no nested iteration constructs, and when expanding a pattern into a concrete trace,
the starred sub-expressions have to be repeated the same number of times; for example,
the pattern ${\vec p}^*{\vec q}^*$ generates the traces $\{ {\vec p}^t{\vec q}^t,\ t\ge 0\}$.
The proof constructs a pattern for every
multi-polynomial computed,  showing that it is realizable. It is interesting that such simple patterns suffice to establish tight lower bounds for all loops in our class.
\item In the upper-bound proof, we describe all  traces by a finite set of \emph{well-typed regular expressions}~\cite{Bojanczyk:DLT:09}.
This elegant tool channels the power of the Factorization Forest Theorem~\cite{Simon:TCS:90}; this theorem
exposes the role of idempotent elements, which is key in our algorithm.
\item Interestingly, the lower-bound proof not only justifies the tightness of our upper bounds, it also justifies the termination of the algorithm and
the application of the Factorization Forest Theorem in the upper-bound proof, because it shows that our abstract multi-polynomials generate a finite monoid. 
\end{itemize}

\noindent
Recall that traces are just sequences of polynomial transitions (Definition~\ref{def:trace}).
Next, we define \emph{concrete traces}, which represent executions on concrete data. The following definition assumes that the steps of the trace are specified by
$\tau$-free multi-polynomials; it suffices for the first part of the proof (Section~\ref{sec-proof-lb}).

\begin{defi}[Concrete traces (unweighted)]
A \emph{concrete trace} corresponding to $\sigma \in \tseqs{\mathcal S}$ is a
path of labeled arcs in the state space $\nats^n$:
\[ \concrete{\sigma} = \trans{{\vec p}_1}{s_{0}}{s_1}\dots\trans{{\vec p}_t}{s_{t-1}}{s_t},\]
where  $s_{i+1} = {\vec p}_{i+1}(s_i)$.
We write $ \tseq{{s_0}}{{\sigma}}{{s_t}}$ to indicate just the initial and final states of a trace (as a special case, the empty trace $\eps$
corresponds to a path with no arcs: $\tseq{s_0}{\eps}{s_0}$).
The set of all concrete traces is denoted $\concrete{\tseqs{\mathcal S}}$.
\end{defi}

\noindent
Note that the semantics of polynomial transitions never block a state. That is, given $\sigma$ and $s_0$ there always is a $\concrete{\sigma}$ starting
with $s_0$.
\paragraph{\textbf{Notations}}
\begin{itemize}
\item
Concatenation of traces $\sigma, \rho$
is written as $\sigma \concat \rho$. For concrete traces, concatenation requires the final state of $\sigma$ to be the initial state of $\rho$,
assuming both are non-empty.
\item
In the proofs, we handle the abstract polynomials computed by the algorithm as if they were concrete polynomials.
This should be understood as an ``implicit cast,'' applying the $\gamma$ conversion function (see Page~\pageref{def:gamma}).
In fact, it is useful to bear in mind that
the correctness claim---namely, Theorem~\ref{thm:sdlp}---is a claim about concrete MPs, relating their numeric values to values generated by program executions.
We mostly omit $\alpha$'s and $\gamma$'s and the reader should interpret a MP as concrete or abstract according to context. For example,
comparison of value $p \ge q$ applies to concrete polynomials, while the statement that $\vec p$ is idempotent indicates that $\vec p$ is
(or should be ``cast'' into) an abstract MP\@.
\item
 We call a $\tau$-MP $\vec p$ $\tau$-closed if ${\vec p}^\tau = {\vec p}$.
\end{itemize}


\subsection{Lower-Bound Correctness for Simple Loops}%
\label{sec-proof-lb}

The key notion in proving that the upper bounds that we compute are tight---equivalently, that they provide lower bounds (up to constant factors) on the
worst-case results---is that of \emph{realizability}. Intuitively, when we come up with a bound, we want to show that there are traces which achieve (realize) this bound.
Importantly, with asymptotic analysis, a bound is not justified, in general, by showing a \emph{single} trace;
what we  need is a \emph{pattern} that generates  arbitrarily long traces.

Formally, we define the class of \emph{patterns} (over a given set of polynomial transitions, $\mathcal S$)
and their associated languages (sets of traces). The following statements define patterns $\pi$ along with corresponding sets of languages, $L(\pi,t)$.
The role of $t$ is related to loop counts; it tells how many times to repeat $\vec a$ in a pattern ${\vec a}^*$.
\begin{itemize}
\item The empty string $\eps$ is also a pattern. It generates the language $L(\eps,t) = \{\eps\}$, consisting of the empty trace.
\item A single MP, ${\vec p}\in \mathcal S$, is a pattern. It generates the language  $L({\vec p},t) = \{ {\vec p} \}$, consisting of a single trace.
\item A concatenation of patterns is a pattern, and $L(\pi_1 \pi_2, t) \eqdef L(\pi_1, t) L(\pi_2, t)$, where the right-hand side applies concatenation to traces. \\
We define concatenation of patterns to be associative, so ${\vec p}({\vec q}{\vec r})$
and $({\vec p}{\vec q}){\vec r}$ are the same pattern, which may be written ${\vec p}{\vec q}{\vec r}$ (this works because concatenation of traces is also associative).
\item If $\pi$ is a pattern, $\pi^*$ is a pattern. However, nested application of the star is not allowed; in other words, $\pi$ is required to be star-free.
We define $L(\pi^*, t) = {L(\pi)}^t$, where $L^0 = \{\eps\}$ and $L^{t+1} = L^t\cdot L$. \\
 Parentheses are used for syntactic disambiguation of the operand of the star, e.g., ${\vec p}{({\vec q}{\vec r})}^*$.
\end{itemize}

\noindent
We use the notation $\pi^n$ as a shorthand for a concatenation of $n$ copies of $\pi$.
The set of traces corresponding to pattern $\pi$, denoted $L(\pi)$,
is defined by $L(\pi) = \bigcup_{t\ge 0} L(\pi,t)$.
 E.g., $L(a^*b^*) = \{ a^n b^n \mid n\ge 1 \}$.
We denote by $\pi(n)$ the result of substituting $n$ for all the stars in $\pi$ (obtaining a star-free pattern).
In particular, $\pi(1)$ is the result of substituting 1 for all the stars in $\pi$. Since nesting of iterative expressions is not allowed,
$|\pi(t)| = \Theta(t)$.

\begin{defi}[Realizability]%
\label{def:realized}
A polynomial transition (PT)  ${\vec p}\in \tppol$  is said to be \emph{realizable} over the given set $\setS\subset \ppol$ if there is a pattern $\pi$
and a constant $0 < c \le 1$,
such that for all $t\ge 1$, for all $\sigma\in L(\pi(t))$, if $\tseq{\vec{x}}{{\sigma}}{\vec{y}}$, then%
\begin{equation} \label{eq:realization}
\vec y \ge  \unimono{\vec p}({\vec x}, t) + c\nunimono{\vec p}({\vec x}, t) \,.
\end{equation}
We say that $\vec p$ is realized by $\pi;c$, or that it is $c$-realizable.   For $\tau$-free MPs, we use the same definition but $\pi$ should not
 include a star, so $t$ may be omitted. Thus a $\tau$-free MP is to be realized by a single abstract trace.
 A set of MPs is called realizable if all its members are.
 \end{defi}

 \begin{exa}
 We reconsider the loop in Example~\ref{ex:motivate-tau}. In Example~\ref{ex:tauisation:1} we have computed from its body the $\tau$-MP
 ${\vec q} = {\vec p}^\tau = \fourtuple{ x_1 + \tau x_3}{ x_2 + \tau x_3+x_4}{ x_3}{ x_3 }$.  We claim that it is realizable. We use the simple pattern $\pi = {\vec p}^*$.
 The final state after $t\ge 1$ iterations (a trace in $L(\pi(t))$) is
 \[
  {\vec y} = \fourtuple{ x_1 + t x_3 }{ x_2 + (2t-1)x_3 + x_4 }{ x_3 }{ x_3 },
\]
which is easy to verify by inspection of the program. We now verify that
\[
{\vec y} \ge \unimono{{\vec q}}({\vec x}, t) + c\nunimono{{\vec q}}({\vec x}, t),
\]
where $c = 1$.
\end{exa}

 Note that realizability is a monotone property in the sense that if ${\vec p} \ge {\vec q}$ and $\vec p$ is realizable, then $\vec q$ is.
 This is natural, since we are arguing that $\vec p$ is a \emph{lower bound}.
 Note also that in contrast to Clause~2 of Theorem~\ref{thm:sdlp}, we pass $t$ rather than $|\sigma|$ as the $\tau$ parameter to the polynomial bound.
 However, since $|\pi(t)| = \Theta(t)$, the inequality in the theorem will be satisfied, just with a different constant factor.%
\footnote{This consideration only applies to the non-linear part $\nunimono{\vec p}$, since $\unimono{\vec p}$ is $\tau$-free anyway.}

\noindent
We introduce a short form that makes inequalities such as~\eqref{eq:realization} easier to manipulate:
\begin{defi}
For any $\tau$-polynomial $p$, and real number $c \le 1$,
\[
 c\mnl p({\vec x},t) \eqdef   \unimono{p}({\vec x}, t) + c\nunimono{p}({\vec x}, t) \,.
\]
We also write $c\mnl {\vec p}$ for component-wise application of ``$c\mnl$'' to a multi-polynomial $\vec p$.
\end{defi}
\noindent
It is useful to be aware of properties of this operation. We leave the verification of the next lemma to the reader.

\begin{lem}%
\label{lem:mnl-properties}
Let $r\in \tpol$, ${\vec p}\in \tppol$, and $c\le 1$. Then we have
\begin{align*}
(c\mnl r ) \circ {\vec p} &\ge c\mnl (r\circ {\vec p}) \\
r \circ (c\mnl {\vec p}) &\ge c^{\deg r}\mnl (r\circ {\vec p}) \,.
\end{align*}
\end{lem}


\begin{lem}
Let ${\vec p},{\vec q} \in \tabsppol$. Then ${\vec p}\circ {\vec q} \ge  {\vec p}\acirc {\vec q}$.
\end{lem}

\noindent
Note that on the left-hand side of the inequality, we compose in $\tppol$, while on the right-hand side, we compose in $\tabsppol$. A more
explicit expression of this claim would be:
\[\gamma({\vec p})\circ \gamma({\vec q}) \ge \gamma ( {\vec p}\acirc {\vec q} ). \]

\noindent
Now this becomes obvious, since both sides of the equality have the same monomials (they are abstractly the same), and
the coefficients on the right-hand side are all 1, while on the left-hand side they are integers.
The fact that both sides have corresponding monomials allows us to strengthen the statement
to the following:
\begin{lem}%
\label{lem:circ-vs-acirc}
Let ${\vec p},{\vec q} \in \tabsppol$ and $c\le 1$. Then $c\mnl({\vec p}\circ {\vec q}) \ge  c\mnl({\vec p}\acirc {\vec q})$.
\end{lem}

 \noindent
 Recall that we aim to prove that all PTs computed by our algorithm are realizable.  We do this by a series of lemmas which shows that realizability is preserved
 by the various constructions in our algorithm. Note that our algorithm produces abstract MPs; so, let us clarify that when we say that ${\vec p}\in\absppol$
 is realizable, Definition~\ref{def:realized} is actually applied to $\gamma({\vec p})$.

\subsubsection{Realizability without generalization}

The realizability of PTs computed in the closure step will follow quite easily using a few simple lemmas.

\begin{lem}%
\label{lem:idppol-realizable}
$\idppol$ is realizable.
\end{lem}

\begin{proof}
$\idppol$ is realized by the empty pattern.
\end{proof}

\begin{lem}
Every member of $\setS$ is realizable.
\end{lem}

\begin{proof}
${\vec p}\in \setS$ is realized by the pattern $\vec p$.
\end{proof}

\begin{lem}%
\label{lem:comp-realizable}
Let $\vec p$, $\vec q$ be realizable $\tau$-MPs. Then ${\vec q}\circ {\vec p}$ is realizable.
\end{lem}

\begin{proof}
Suppose that $\vec p$ is realized by $\pi_1;c_1$ and $\vec q$ by $\pi_2;c_2$.  We claim that
${\vec q}\circ {\vec p}$ is realized by $\pi_1 \pi_2 ; c'$ for some $c'$.  Consider a concrete transition sequence $\widetilde \sigma$ with
$\sigma\in L((\pi_1\pi_2)(t))$, it has the form $\tseq{\vec x}{\sigma_1}{\vec y}\tseq{\:}{\sigma_2}{\vec z}$ with
$\sigma_i \in L(\pi_i(t))$.
We have:
%
\begin{align*}
\vec y & \ge  c_1\mnl{\vec p}({\vec x}, t) \\
\vec z & \ge  c_2\mnl{\vec q}({\vec y}, t),
\end{align*}
consequently, using Lemma~\ref{lem:mnl-properties}:
\[
 \vec z \ge c_1^d c_2\mnl({{\vec q}\circ {\vec p}})({\vec x}, t)
\]
where $d$ is the highest degree of any component of $\vec q$.
\end{proof}

\begin{lem}\label{lem:acomp-realizable}
Let ${\vec p}, {\vec q}\in \tabsppol$ be realizable. Then ${\vec q}\acirc {\vec p}$ is realizable.
\end{lem}

 \begin{proof}
 This follows from the previous lemma and Lemma~\ref{lem:circ-vs-acirc}.
 \end{proof}

\begin{cor}
Let $T$ be a set of realizable abstract $\tau$-MPs.  Then every member of $\closure{T}$ is realizable.
 \end{cor}

\subsubsection{Dependence graphs and neat MPs}

We are still missing a realizability lemma for the generalization operation.
As a preparation for this proof, we
 introduce the dependence graph of a PT and state an important structural property of dependency graphs associated with idempotent transitions.
 This leads to the definition of the class we call {\em neat MPs}.

\begin{defi}
Let ${\vec p} \in \tppol$. Its \emph{dependency graph} $G({\vec p})$ is
a directed graph with node set $[n]$. The graph includes an arc $i\to j$ if and only if
${\vec p}[j]$ depends on $x_i$.
\end{defi}

\noindent
Intuitively, $G({\vec p})$ shows the data flow in the transition $\vec p$. It is easy to see that paths in the graph correspond to data-flow
effected by a sequence of  transitions.  For instance, if we have a path $i\to j \to k$ in $G({\vec p})$, then $x_i$ will appear in the expression
$({\vec p}\circ {\vec p})[k]$.

\begin{lem}\label{obs:no-doublin}
Suppose that the SDL $\setS$ is polynomially bounded.
Then for all ${\vec p}\in \tppol$ which is realizable over $\setS$,  ${\vec p}[i]$  has no monomial divisible by $x_i$ other than $x_i$.
\end{lem}

\begin{proof} Suppose that $\vec p$ has such a monomial, then it has the form ${\mathfrak m} x_i$.
If $\mathfrak m$ is $\tau$-free, then
 starting with a state in which all variables have
values greater than 1,
and repeatedly executing a trace that realizes $\vec p$,
we clearly obtain an exponential growth of the value in $\X_i$, contradicting the assumption.
In the case that $\tau$ occurs in $\mathfrak m$, suppose that $\vec p$ is realized by $\pi;c$. Choose $t \ge \max(c^{-1}, 2)$.
Then repeating the trace $\pi(t)$ creates an exponential growth in $\X_i$.
\end{proof}

\begin{lem}%
\label{lem:nearly-dag}
Suppose that ${\vec p}\in \tppol$ is realizable, and $\alpha({\vec p})$ is
idempotent. Assuming that the loop under analysis is polynomially bounded,
 $G({\vec p})$ does not have any simple cycle longer than one arc. In other words,
it consists of a directed acyclic graph (DAG) plus some self-loops.
\end{lem}

\begin{proof}
Assume, to the contrary, that $G({\vec p})$ does have a cycle $i\to \dots \to k \to i$, where $k\ne i$. Let $r$ be the length of the cycle. Then
\begin{enumerate}
\item\label{stp1}
 ${\vec p}^{(r)}[i]$ depends on $x_i$; we can write ${\vec p}^{(r)}[i] = x_i + p(\vec x, \tau)$ (the occurrence of $x_i$ must be in a linear monomial
with a concrete coefficient of 1, otherwise iteration would cause exponential growth).
By the remark after Definition~\ref{def:itcomp}, $\alpha({\vec p}^{(r)})[i] = ({(\alpha({\vec p}))}^{\acirc (r)})[i]$,
which, by idempotence, equals
$\alpha({\vec p})[i]$. Reading backwards, $\alpha({\vec p}[i]) = \alpha({\vec p}^{(r)})[i] = \alpha(x_i) + p(\vec x, \tau)$, so
${\vec p}[i]$ also has the form $x_i + q(x,\tau)$ (where $q$ may differ from $p$).
\item By assumption, ${\vec p}[i]$ depends on $x_k$. Since $k\ne i$, we conclude that $q$ depends on $x_k$.
Also by the assumption of the cycle in $G({\vec p})$,  ${\vec p}^{(r-1)}[k]$ depends on $x_i$.
This shows that $q\circ {\vec p}^{(r-1)}$ depends on $x_i$.
\item
Now, \[
  {\vec p}^{(r)}[i] = ({\vec p}\circ {\vec p}^{(r-1)})[i] =  {\vec p}[i]\circ {\vec p}^{(r-1)} =
       (x_i + q)\circ {\vec p}^{(r-1)} = {\vec p}^{(r-1)}[i] + (q\circ {\vec p}^{(r-1)})  \,.
\]
We argue that the last expression has at least \emph{two} occurrences of $x_i$.
First, by the same token as (\ref{stp1}), $ {\vec p}^{(r-1)}[i]$ has an occurrence of $x_i$.  Secondly,
$q\circ {\vec p}^{(r-1)}$ also has one.
\item Thus ${\vec p}^{(r)}$ generates exponential growth when iterated, after all.
\end{enumerate}
We conclude that a cycle as assumed cannot exist.
\end{proof}

The realizability lemmas are intended to be applied under the assumption of Theorem~\ref{thm:sdlp}, namely
 that our loop is polynomially bounded; therefore we can rely on the properties guaranteed by Lemmas~\ref{obs:no-doublin} and~\ref{lem:nearly-dag}.
We focus on a particular idempotent realizable AMP $\vec p$.
Then $G({\vec p})$, with self-loops removed, is a DAG\@. We assume, w.l.o.g., that the variables
are indexed in an order consistent with $G({\vec p})$, so that if $x_i$ depends on $x_j$ then $j\le i$.
We shall refer to an (abstract) MP satisfying the properties in Lemma~\ref{obs:no-doublin}, Lemma~\ref{lem:nearly-dag} and with variables
so re-indexed (if necessary) as being \emph{neat}.\footnote{Readers who like linear algebra may draw some intuition about neat MPs
from thinking about triangular matrices whose diagonal elements are in $\{0,1\}$.}

\begin{defi}
${\vec p}\in \tppol$ is called \emph{neat} if  (a) $G({\vec p})$, with self-loops removed, is a DAG;\@
(b) for all $i$,  ${\vec p}[i]$  has no monomial divisible by $x_i$ except (possibly) $x_i$.
\end{defi}

Note that $\vec p$ is not required to be idempotent; neat MPs that are not necessarily idempotent come up in the upper-bound proof.
In the rest of this subsection we study neat idempotent AMPs,
establishing (as Corollary~\ref{cor:lb-nsd} below) a property important for the subsequent realizability lemma.
This property involves the $n$th iterate ${\vec p}^{(n)}$.
It is not hard to prove that $({\vec p}^{(n)})(s') \ge  {\vec p}(s)$ if
$s'[i] \ge s[i]$ for all $i$. We skip the proof, however, since
we can establish a sharper result, for states $s$, $s'$ where $s'[i]\ge s[i]$ is only asserted for $i\in \sd{\vec p}$ (i.e., $x_i$ is self-dependent in $\vec p$).

We introduce the following notation: for ${\vec p}\in \tabsppol$, $\chi_{\vec p}(i)$ is $1$ if $i\in \sd{\vec p}$ and $0$ otherwise.

\begin{lem}%
\label{lem:nsd-induction}
Let ${\vec p}\in \tabsppol$ be idempotent and neat. Let $s,s'\in \nats^n$ be any state vectors such that
for $i\in \sd{\vec p}$, $s'[i] \ge s[i]$ (while for $i\notin \sd{\vec p}$ no relation is asserted).  Then for all $\ell\le n$, 
\( {\vec q}_{\ell} \eqdef {(c\mnl {\vec p})}^{(\ell)} \)
has the following property: for all $i\le \ell$, for all $t\ge 0$,
\[
  {\vec q}_{\ell}[i](s',t) \ge (c^{{(d+1)}^i}\mnl {\vec p})[i](s,t) + (s'[i]-s[i])\cdot \chi_{\vec p}(i),
\]
where $d$ is the maximum degree in $\vec p$.
\end{lem}


\noindent
The following fact is useful for proving the lemma; it follows from Lemma~\ref{lem:mnl-properties}.

\begin{fact}%
\label{lem:inductive-mnl}
Let ${\vec p} \in \tppol$, ${\vec q}\in \tppol$.  Suppose that for all  $i<\ell$, we have ${\vec q}[i] \ge c^{{(d+1)}^i} \mnl {\vec p}[i]$, were $d$
is the maximum degree in $\vec p$.
And suppose that $r$ is a $\tau$-polynomial depending only on $x_1,\dots,x_{\ell-1}$, also of degree at most $d$.
Then $(c\mnl r) \circ {\vec q} \ge c^{{(d+1)}^\ell}\mnl (r\circ {\vec p})$.
\end{fact}

\begin{proof}[Proof of Lemma~\ref{lem:nsd-induction}]
We use induction on $\ell$. Since $\vec p$ is neat we know that ${\vec p}[1] = x_1$.
Thus
\[
  (c\mnl{\vec p})[1](s',t) = s'[1]  = s[1] + (s'[1]-s[1])\cdot \chi_{\vec p}(1)  = c^{d+1}\mnl{\vec p}[1](s,t) + (s'[1]-s[1])\cdot \chi_{\vec p}(1)  .
\]

Next, for  $\ell > 1$, consider ${\vec q}_{\ell} = (c\mnl {\vec p}) \circ {{\vec q}_{\ell-1}}$.
If $i\le \ell$ and $x_i$ is not self-dependent in $\vec p$, then ${\vec p}[i]$ only depends on variables $x_j$ with $j<i$, so:
\begin{align*}
  {\vec q}_{\ell}[i](s',t) &= (c\mnl {\vec p}[i] \circ {{\vec q}_{\ell - 1}})(s',t) \\
 & \ge (c\mnl {\vec p}[i] \circ c^{{(d+1)}^{i-1}}\mnl {\vec p})(s,t) && \text{by IH} \\
 & \ge (c^{{(d+1)}^i} \mnl ({\vec p}[i] \circ {\vec p}))(s,t) && \text{by Fact~\ref{lem:inductive-mnl}}\\
 &  \ge (c^{{(d+1)}^i}\mnl {\vec p}[i])(s,t)  && \text{by idempotence} \\
&= (c^{{(d+1)}^i}\mnl {\vec p}[i])(s,t)  + (s'[i]-s[i])\cdot \chi_{\vec p}(i)
\end{align*}
where the next-to-last step uses the idempotence of $\vec p$ in $\tabsppol$ and Lemma~\ref{lem:circ-vs-acirc}.

If $x_i$ is self-dependent in $\vec p$, then  ${\vec p}[i] = x_i + r(\vec x, \tau)$ where $r$ only depends on variables $x_j$ with $j<i$. Now by IH,
and Fact~\ref{lem:inductive-mnl}:
\begin{align*}
  {\vec q}_{\ell}[i](s',t) &= (c\mnl {\vec p}[i] \circ {{\vec q}_{\ell - 1}})(s',t) \\
&= {{\vec q}_{\ell - 1}}[i](s',t) + (c\mnl r \circ {{\vec q}_{\ell - 1}})(s',t) \ge s'[i] + (c\mnl r \circ (c^{{(d+1)}^{i-1}}\mnl {\vec p}))(s,t)  \\
&=  s[i] + (s'[i]-s[i])\cdot \chi_{\vec p}(i)   + (c\mnl r \circ (c^{{(d+1)}^{i-1}}\mnl {\vec p}))(s,t)  \\
&= (c\mnl {\vec p}[i] \circ (c^{{(d+1)}^{i-1}}\mnl {\vec p}))(s,t) + (s'[i]-s[i])\cdot \chi_{\vec p}(i)  \\
&\ge c^{{(d+1)}^i} \mnl ( {\vec p}[i] \circ {\vec p})(s,t) + (s'[i]-s[i])\cdot \chi_{\vec p}(i)  \\
&\ge  (c^{{(d+1)}^i}\mnl {\vec p}[i])(s,t) + (s'[i]-s[i])\cdot \chi_{\vec p}(i)  \,. \qedhere
\end{align*}
\end{proof}

\noindent
Letting $\ell=n$ in this lemma we obtain:

\begin{cor}%
\label{cor:lb-nsd}
 Let ${\vec p}\in \tabsppol$ be idempotent and neat.   Let $s,s'\in \nats^n$ be any state vectors such that
for $i\in \sd{\vec p}$, $s'[i] \ge s[i]$.  Then, for all $i\le n$ and all $t$, ${(c\mnl {\vec p})}^{(n)}[i](s',t) \ge  (c^{{(d+1)}^n}\mnl {\vec p})[i](s,t)+ (s'[i]-s[i])\cdot \chi_{\vec p}(i)$.
\end{cor}

\subsubsection{Realizability lemma for generalization}

This is what we have been preparing for.

\begin{lem}%
\label{lem:tau-realizable}
Let $\vec p$ be a realizable, idempotent and neat abstract $\tau$-MP\@.
Then ${\vec p}^\tau$ is realizable.
\end{lem}

Note that if $\vec p$ is $\tau$-closed, then ${\vec p}^\tau = {\vec p}$ and the statement is trivial.
But if $\vec p$ is not $\tau$-closed, then we have to construct a pattern that causes additive terms to accumulate in the self-dependent
variables in order to justify the replacement of sums  $x_i + {\vec p}[i]''$ by $x_i  + \tau{\vec p}[i]''$.
The interesting case is when $\vec p$ is not $\tau$-free:
then $\pi$ already has stars, however nested stars are not allowed in a pattern, so we cannot iterate $\pi$.
The solution is to use ${(\pi(1))}^*$.  However, in reducing $\pi$ to
the star-free $\pi(1)$ we pull the rug under $\tau$-monomials already present in $\vec p$. Therefore we use a pattern that includes both ${(\pi(1))}^*$ and
$\pi$;  actually, not $\pi$ but ${\pi}^n$, in order to use Corollary~\ref{cor:lb-nsd}. Note that realizability of $\vec p$ by $\pi$ means that
$\pi(1)$ realizes ${\vec p}[1/\tau]$ (which is $\vec p$ with all $\tau$'s erased). The following observation is useful:
\begin{obs}
If ${\vec p}\in \tabsppol$ is idempotent, so is ${\vec p}[1/\tau]$; and similarly for neat.
\end{obs}
\newcommand{\pmt}{\ensuremath{{\vec p}[1/\tau]}}

\begin{proof}[Proof of Lemma~\ref{lem:tau-realizable}]
 Let $\pi ; c$ realize $\vec p$. Note that $\pi$ will include stars iff $\vec p$ includes $\tau$'s. We are handling the general case so we treat
 $\vec p$ as a $\tau$-MP;\@ the arguments also apply to the case that it is $\tau$-free.
So, our assumption is that:
\begin{equation} \label{eq:ass2}
\text{if }  \tseq{\vec{x}}{{\pi(t)}}{\vec{y}} \text{   then  } \vec y \ge c\mnl {\vec p}(\vec x, t) \,.
\end{equation}

We will show that ${\vec p}^\tau$ is realized by the pattern ${(\pi(1))}^* \pi^n$.
Now,  traces in $L({(\pi(1))}^{*} \pi^n )$ have the form $\sigma_t =  {(\pi(1))}^t \pi^n(t)$.
We look at concrete traces \[\tseq{s_{0}}{\pi(1)}{s_1}\tseq{\cdots}{\pi(1)}{s_{t}}\tseq{\:}{\pi^n(t)}{s_{t+1}} \,.\]

We first prove by induction on $t$,
for $x_i$  self-dependent in $\vec p$, and $t\ge 0$,
\begin{equation} \label{eq:add-t}
s_{t}[i] \ge  s_{0}[i]  + t (c\mnl {\vec p}[i]''(s_0))
\end{equation}
(note that the omission of the $\tau$ parameter in ${\vec p}[i]''(s_0)$ reminds us that ${\vec p}[i]''$ is $\tau$-free).
For $t=0$,
the statement is trivial.
For the inductive case,
\smallskip
\[
\begin{alignedat}{2}
s_{t}[i] &\ge c\mnl {\vec p}[i](s_{t-1}, 1) &\quad&\  \notag\\
&\ge s_{t-1}[i] +  c\mnl ( {\vec p}[i]'(s_{t-1}, 1) + {\vec p}[i]''(s_{t-1}) + {\vec p}[i]'''(s_{t-1},1)) & 
                 &  \text{by structure of $\vec p[i]$} \notag\\
&\ge s_{t-1}[i] +   c\mnl {\vec p}[i]''(s_{t-1}) \notag\\
&\ge s_0[i] +  (t-1)(c\mnl {\vec p}[i]''(s_0)) + c\mnl {\vec p}[i]''(s_{t-1}) & & \hspace{-1cm} \text{by IH} \notag\\
&\ge s_0[i] + (t-1) (c\mnl {\vec p}[i]''(s_0))
    +  c\mnl {\vec p}[i]''(s_0) &
             &   \hspace{-1cm}\text{(s.d.~variables are non-decreasing)} \notag\\
&\ge s_0[i] + t {\vec p}[i]''(s_0) \notag
\end{alignedat}
\]
\noindent
Now, for $i\in \sd{\vec p}$,
\begin{alignat*}{2}
 s_{t+1}[i] &\ge {(c\mnl {\vec p})}^{(n)}[i](s_{t}, t) &\quad& \\
 &\ge (c^{{(d+1)}^n}\mnl {\vec p})[i](s_0,t) + (s_t[i]-s_0[i]) && {\text{by Cor.~\ref{cor:lb-nsd}}} \\
  & = s_{0}[i] + c^{{(d+1)}^n}\mnl (t{\vec p}[i]' (s_{0},t) + {\vec p}[i]''(s_{0}) + {\vec p}[i]'''(s_{0}, t) )  + (s_t[i]-s_0[i]) && \text{by Def.~\ref{def:MPnotation}} \\ 
  & \ge s_{0}[i] +  c^{{(d+1)}^n}\mnl (t{\vec p}[i]' (s_{0},t) + {\vec p}[i]''(s_{0}) + {\vec p}[i]'''(s_{0}, t) ) + t {\vec p}[i]''(s_0) && \text{by~\eqref{eq:add-t}} \\ 
 & \ge s_{0}[i] +  c^{{(d+1)}^n}\mnl (t{\vec p}[i]' (s_{0},t) + t{\vec p}[i]''(s_{0}) + {\vec p}[i]'''(s_{0}, t) ) \\ 
  &  = c^{{(d+1)}^n}\mnl {\vec p}^\tau[i](s_0, t) \,.
 \end{alignat*}
\noindent
For $i\notin \sd{\vec p}$,
\begin{alignat*}{2}
 s_{t+1}[i] &\ge {(c\mnl {\vec p})}^{(n)}[i](s_{t}, t) &\quad& \\
 &\ge (c^{{(d+1)}^n}\mnl {\vec p})[i](s_0,t) + (s_t[i]-s_0[i])\cdot \chi_{\vec p}(i) && {\text{by Cor.~\ref{cor:lb-nsd}}} \\
 & = (c^{{(d+1)}^n}\mnl {\vec p})[i](s_0,t) \,.
 \end{alignat*}
\noindent
We conclude that the pattern ${(\pi(1))}^*\pi^n$ realizes ${\vec p}^\tau$.
\end{proof}

\noindent
We can now complete our lower-bound correctness proof:
\begin{thm}%
\label{thm:SDL-lb}
Consider any set of MPs $\setS$.
Then the set of $\tau$-MPs returned by \procSDL{}$(\setS)$ is realizable over $\setS$.
\end{thm}
\begin{proof}
This follows  by induction on the number of operations used to construct each resulting $\tau$-MP, being either composition steps or generalization
steps (justified, respectively, by Lemma~\ref{lem:comp-realizable} or Lemma~\ref{lem:tau-realizable}).
\end{proof}

\noindent
This has a crucial corollary:

\begin{cor}%
\label{cor:termination}
Consider any set of MPs $\setS$, that represent a SDL in which all variables are polynomially bounded.
Then the set of $\tau$-MPs generated by \procSDL{}$(\setS)$ is finite, and therefore obtained in a finite number of steps.
\end{cor}

\begin{proof}
Since we are working within $\tabsppol$, it is clear that if the set of $\tau$-MPs computed by the algorithm is infinite, unbounded exponents must
appear. Since all these MPs are realizable, this indicates that some variable cannot be bounded polynomially over all executions of the loop.
 Taking the contrapositive, we conclude
that if all variables are known to be polynomially bounded, the set is finite.
\end{proof}

\subsection{Upper-Bound Correctness}%
\label{sec-proof-ub}

The upper-bound correctness establishes a correspondence between the set of AMPs computed by our algorithm and a set of concrete polynomials that
actually provide upper bounds for all loop executions (as in Theorem~\ref{thm:sdlp}, Clause~1).
Implicitly, this proof provides \emph{an algorithm to compute concrete upper bounds}, while ensuring that the abstractions of these bounds
remain within our set of AMPs. But we do not attempt to explicate the algorithm as such, since our main contribution is the algorithm already given.
Firstly, we have to provide a finer definition of ``upper bound'' in terms of \emph{weighted traces}.
 The main technical part of the section defines a special class of $\tau$-MPs,
called \emph{iterative}, and shows some properties of the class. Finally, the main part of the proof uses a corollary of Simon's
\emph{Forest Factorization Theorem} to construct the desired upper bounds.
The application of this theorem is justified by Corollary~\ref{cor:termination}.
Before diving into upper bounds, we give a useful definition and some related lemmas.

 \begin{defi}%
 \label{def:join}
 For multiivariate polynomials $p,q$ we define their \emph{join} $p \join q$ as the polynomial
obtained by setting the coefficient of every monomial
to the largest of its coefficients in $p$ and in $q$.
\end{defi}


\noindent
For example: $(2x_1+x_1 x_2) \join (x_1+x_2+3x_1 x_2) = 2x_1 + x_2 + 3x_1x_2 \,$.

\bigskip\noindent
It is easy to see that $\join$ is the join (least upper bound) operation in a natural partial order on polynomials, which we denote by $\lepoly$.
Note that this order is ``syntactic'' and is a part of the ``semantic'' relation
$p \le q$ defined  by treating them as functions over $\nats^n$. To see that the relations do not coincide, consider $xy$ versus $x^2+y^2$.
With multi-polynomials, we apply the relation component-wise.

\begin{lem}%
\label{lem:lepoly-and-circ}
Suppose that ${\vec p} \lepoly {\vec q}$ and ${\vec r} \lepoly {\vec s}$. Then
${\vec p} \circ {\vec r} \lepoly {\vec q} \circ {\vec s}$.
\end{lem}

\begin{proof}
Every monomial  $\mathfrak m$ of ${\vec p} \circ {\vec r}$ is the result of substituting a certain monomial of ${\vec r}[i]$ for every occurrence
of $x_i$ in $\vec p$ (where $x_i^d$ counts as $d$ occurrences). As the same monomials (up to coefficients)
appear in, respectively, ${\vec s}[i]$ and $\vec q$, we conclude that $\mathfrak m$ (times some constant) appears in ${\vec q}\circ{\vec s}$.
\end{proof}

Using the join operator, we thus have: ${\vec p} \circ ({\vec r} \join {\vec s})  \gepoly ({\vec p} \circ {\vec r}) \join ({\vec p} \circ {\vec s})$.

\begin{lem}%
\label{lem:afplusbg}
For any pair of polynomials $f,g$, and numbers $a,b\ge 0$, 
\[
 af + bg \le (a+b) (f \join g) \,.
\]
\end{lem}

\begin{proof}
 Partition the monomials of $f$ $\join$ $g$ into three groups: those that appear only in $f$; they are multiplied by $a$ on the LHS
and by $(a+b)$ on the RHS\@.   Those that appear only in $g$; they are multiplied by $b$ on the LHS
and by $(a+b)$ on the RHS\@.
For a monomial $\mathfrak m$ that appears in both polynomials---possibly with different coefficients, say $c$ in $f$ and $d$ in $g$,
we find $ac{\mathfrak m} + bd{\mathfrak m}$ in the LHS, and $(a+b)\max(c,d){\mathfrak m}$ in the RHS\@.
\end{proof}

\paragraph{\bf Weighted traces and upper bounds}

We extend our notion of traces by defining \emph{weighted traces}.
The need for this definition arises because, in the proof, we argue about a decomposition of an actual trace into segments, and we need
to abstract from the intermediate states within a segment.  The number of actual transitions in the segment matters, however,
and will be represented by its weight.
Thus a weighted step, denoted $\wtrans{{\vec p}}{s}{s'}{w}$, indicates that state $s$ has evolved into step $s'  = {\vec p}(s,w)$ in $w$ actual
transitions, for an integer $w\ge 1$.
A weighted concrete trace $\sigma$ is composed of such steps in the same way as an ordinary trace is; a weighted abstract trace specifies the bounds and
weights but leaves the states unspecified, e.g., $({\vec p}_1|w_1)({\vec p}_2|w_2)({\vec p_3}|w_3)$.
The total weight of a trace $\sigma$ is denoted by $\wgt{\sigma}$ and
 calculated by adding up the weights of the steps. Thus we have
$\wgt{\sigma\rho} = \wgt{\sigma} + \wgt{\rho}$.
Note that since the MPs that label a weighted step are, in general, upper bounds obtained in previous analysis steps, the value
$s' = {\vec p}(s,w)$ represents, in general, a bound and not a state that is actually reachable; however,
we argue that due to the monotonicity of the functions computable in our language, we do not lose soundness by considering $s'$ instead of a set of values
bounded by $s'$.

We denote the set of weighted traces over $\setS$ by $\wtseqs{\setS}$.
Note that weights are associated with abstract transitions and all concretizations of a weighted trace have the same weight.


\begin{defi}[bound for a trace]
Let $\sigma\in \concrete{\wtseqs{\mathcal S}}$.
We say that $\sigma$ admits
a $\tau$-polynomial $p$ as an upper bound on variable $j$, or that $p$ bounds variable $j$ in $\sigma$,
if the following holds:
\begin{equation} \label{eq:description}
  \tseq{\vec{x}}{\sigma}{\vec{y}} \land
  {t \ge \wgt{\sigma}}
  \ \Longrightarrow\
 y_j\le p(\vec{x}, t)
 \,.
\end{equation}
%
%
%
If $\vec p$ bounds \emph{all} variables in a concrete weighted trace $\tseq{\vec x}{\sigma}{\vec y}$,
we say that $\vec p$ bounds this trace. If $\vec p$ bounds all concretizations of $\sigma$, we say it bounds $\sigma$.
\end{defi}

\noindent
The following lemma follows from the monotonicity of all our polynomials.

\begin{lem}%
\label{lem:comp-ub}
If $\vec p$ bounds $\sigma$ and
   $\vec q$ bounds $\rho$, then
   ${\vec q}\circ{\vec p}$ bounds $\sigma\rho$.
\end{lem}
 \begin{proof}
Let $t_s = \wgt{\sigma}$, $t_r = \wgt{\rho}$. Consider a concrete run of $\sigma\rho$ and name the states  thus:
 $\tseq{\vec x}{\sigma}{\vec y}\tseq{\:}{\rho}{\vec z}$.  Then
\(
  z_i  \le  {\vec q}[i](\vec y, t_r)
\)
(by assumption) and for any $j$,
\(
  y_j \le  {\vec p}[j](\vec x, t_s)
\).
It clearly follows that
\(
 z_i \le  ({\vec q}[i] \circ \vec p)(\vec x, t_s+t_r)
\).
Note that $t_s+t_r$ is precisely $\wgt{\sigma\rho}$.
\end{proof}

\paragraph{\bf Iterative MPs and upper bounds for sequences of similar MPs}

\begin{defi}
$\vec p \in \tppol$ is  called \emph{iterative} if all its entries depend only on self-dependent variables. Moreover, each ${\vec p}[i]$ depends
only on variables with indices $j\le i$.
\end{defi}

\noindent
\begin{exa}
 $\tuple{x_1, x_2+x_1, x_2 x_1}$ is iterative, while $\tuple{x_1, x_1, x_3+x_2+x_1}$ is not
(note that the last MP is idempotent when abstracted to $\absppol$, which demonstrates that this property does not imply iterativity; the converse does not hold, either).
\end{exa}
We abbreviate ``iterative multi-polynomial'' to IMP\@.
We next give some technical definitions and results regarding IMPs. The upshot of this part is the computation of tight upper bounds for the end-state
of certain traces composed of IMPs.
We now attempt to give some intuition.  Consider a sequence of MPs
which are idempotent in their abstract form.  Thus if we compose their abstractions in $\absppol$ we just get the same AMP again.  This is a nice
situation which suggests that we can extrapolate the bounds for a single step to bounds for any number of steps.
But there may be
some values that grow and we will have to account for this in our upper bounds.  Consider the following MP:\@
\[ \tuple{x_1+x_2+x_3, \ x_3, \ x_3} \]
Note that it is idempotent in $\absppol$.  However, in concrete computation, the increments accumulate.  But hastily changing the first component
to $x_1 + \tau(x_2+x_3)$ overshoots the upper bound. The correct result is $x_1 + x_2 + x_3 + \tau(2 x_3)$.  We find this result by extracting an iterative
MP from the given set of MPs.  In the above example, the extracted IMP will be
$ \tuple{x_1 + 2x_3, \ x_3, \ x_3} $; we will see later why.

Moving to more precise details, we are going to consider sequences of \emph{similar} MPs, where similarity means that we get the same
MP if we ignore the coefficients as well as $\tau$ symbols, as expressed by the following set of definitions.

\begin{defi}
Let $\vec p \in \tppol$. We define $\tra{\vec p}$ to be $\alpha({\vec p}[1/\tau])$.
\end{defi}

\noindent
For example, $\tra{\tuple{x_1+\tau x_2, x_1^2 + \tau x_1^2}} = \tuple{x_1+x_2, x_1^2}$.

\begin{defi}
We call ${\vec p}$, $\vec q$ \emph{similar} if $\tra{\vec p} = \tra{\vec q}$.
\end{defi}

\begin{obs}%
\label{obs:tra-composition}
Suppose that $\tra{\vec p} = \tra{\vec q}$. Then for any $\vec t$, $\tra{{\vec p}\circ {\vec t}} = \tra{{\vec q}\circ {\vec t}} =
\tra{\vec p} \acirc \tra{\vec t}$.
\end{obs}

Whenever an idempotent MP is considered, we also assume it is realizable, which in turn allows us to assume that
 its variables are indexed in topological order, as in the previous section, making it neat.
One could worry that the assumptions of topological order might contradict each other in an argument that involves
 several different idempotents,
but this is not a problem since all the development below only concerns a set of similar MPs,
so they agree on the topological indexing. In such a discussion we may say, for example, that some $x_i$ is self-dependent, without specifying
which MP of the set is concerned.

\begin{defi}
Let $\vec p \in \tppol$ be such that $\tra{\vec p}$ is neat. Define its \emph{self-dependent cut} $\sdpart{\vec p}$ as follows:
for all $i$ self-dependent in $\vec p$, $\sdpart{\vec p}[i] = x_i$. For all other $i$, $\sdpart{\vec p}[i] = {\vec p}[i]$.
\end{defi}


\begin{lem}%
\label{lem:sdm-property}
Let ${\vec p}_1,\dots,{\vec p}_n$ be similar $\tau$-MPs such that $\tra{{\vec p}_1}=\tra{{\vec p}_2}=\dots$ is neat.  Then for all $\ell\le n$,
\( {\vec q}_\ell \eqdef {\vec p}_1 \circ \sdpart{{\vec p}_2}\cdots \circ \sdpart{{\vec p}_\ell} \)
has the following property: for all $i$, ${\vec q}_\ell[i]$ only depends on variables $x_j$ which are either self-dependent
 or else $j \le i-\ell$.
\end{lem}

\begin{proof}
We use induction on $\ell$. For $\ell=1$, the result is immediate, by virtue of the topological indexing.
  Next, for  $\ell > 1$, consider ${\vec q}_{\ell-1} = {\vec p}_1 \circ \cdots \circ \sdpart{{\vec p}_{\ell-1}}$.
Now ${\vec q}_\ell = {\vec q}_{\ell-1} \circ \sdpart{{\vec p}_\ell}$.
By IH, ${\vec q}_{\ell-1}[i]$ only depends on variables $x_j$ which are either self-dependent or have
$j\le i - \ell + 1$.   In the case that $x_j$ is self-dependent,
then $\sdpart{{\vec p}_\ell}[j] = x_j$, so ${\vec q}_\ell[i]$ also depends on $x_j$ which agrees with the lemma.  In the case that $x_j$ is not
self-dependent,   $\sdpart{{\vec p}_\ell}[j]$ must depend on variables $x_k$ with $k \le j-1 \le i - \ell$.
\end{proof}

\noindent
This lemma is actually formulated for induction; what we are really interested in is the MP denoted by ${\vec q}_n$, for which we  introduce
a special notation while enunciating the following corollary:

\begin{cor}
Let ${\vec p}_1,\dots,{\vec p}_n$ be  similar $\tau$-MPs such that $\tra{{\vec p}_1}=\tra{{\vec p}_2}=\dots$ is neat. Then
\( \sdm{{\vec p}_1\cdots  {\vec p}_n} \eqdef {\vec p}_1 \circ \sdpart{{\vec p}_2}\cdots \circ \sdpart{{\vec p}_n} \)
is iterative.
\end{cor}

\noindent
\emph{Notation}: if $S$ is a set of neat, similar $\tau$-MPs as above, we let
\[ \sdc{S} \eqdef \{ \sdm{{\vec p}_1\:\dots\:{\vec p}_{n}} \, \mid \, {\vec p}_1,\dots,{\vec p}_{n} \in S \} \, . \]

Note that the set $\sdc{S}$ consists of IMPs (by the above corollary) and it is easy to see that they are neat and similar to each other.
Thus, we have  a way to extract a set of IMPs from a certain set of similar MPs.  We will next move to the way in which we compute upper bounds
for sequences of similar MPs using these IMPs.
Recall that a $\tau$-MP is \emph{$\tau$-closed} if ${\vec p}^\tau = {\vec p}$.  In particular, $\tau$-MPs that result directly from generalization
(i.e., ${\vec q}^\tau$ for some $\vec q$) are $\tau$-closed.
We next define a special composition operator for $\tau$-closed IMPs, and show an interesting property that it has.
Note that, using the notation of Definition~\ref{def:MPnotation} (Page~\pageref{def:MPnotation}),
if ${\vec p}$ is a $\tau$-closed IMP then its self-dependent entries have the form
$ {\vec p}[i] = x_i + \tau {\vec p}[i]' $.

\begin{defi}[$\tau$-absorbing composition]
We introduce a non-standard composition operation on similar, $\tau$-closed IMPs, denoted by the operator $\taucomp$ (differing
from both ordinary composition $\circ$ and abstract composition $\acirc$).
Specifically, $({\vec q}\taucomp {\vec p})[i]$ is defined as follows:
\begin{enumerate}
\item If $i\in \sd{\vec q} = \sd{\vec p}$, then
$({\vec q}\taucomp {\vec p})[i] \eqdef  {\vec p}[i] \join  \tau({\vec q}[i]' \circ {\vec p}) = x_i  +  \tau ({\vec p}[i]' \join  ({\vec q}[i]' \circ {\vec p}))$.
\item Otherwise, $({\vec q}\taucomp {\vec p})[i] \eqdef ({\vec q}[i]\circ {\vec p})$.
\end{enumerate}
\end{defi}
\noindent
Observe that the difference between ${\vec q}\taucomp {\vec p}$ and ${\vec q} \circ {\vec p}$ is only in the numeric coefficient of some monomials.
In particular, they have the same abstraction: $\alpha({\vec q}\taucomp {\vec p}) = \alpha({\vec q} \circ {\vec p}) = \alpha({\vec q})\acirc \alpha({\vec p})$.

\begin{exa}%
\label{ex:tauproduct}
Consider the following $\tau$-polynomial:
\[
\begin{array}{lrlllll}
& {\vec p} = \langle & x_1, & x_2+\tau x_1, & x_3 + \tau x_2 + \tau^2 x_1,  & x_4 \rangle \\[1\jot]
\llap{\text{Then:}} \\
& {\vec p}\taucomp {\vec p} = \langle & x_1, & x_2+\tau x_1, & x_3 + \tau x_2 + \tau^2 x_1, & x_4 \rangle & = {\vec p} \\[1\jot]
\llap{\text{while:}} \\
& {\vec p}\circ {\vec p} = \langle & x_1, & x_2+ 2\tau x_1, & x_3 + 2\tau x_2 + 3\tau^2 x_1, & x_4 \rangle .
\end{array}
\]
\end{exa}

Thus, the operation
keeps (some) coefficients from growing when we compose $\tau$-MPs.  We will see next that this has the desirable result that a sequence of
``$\taucomp$ powers'' $\setT^{\taucomp (i)}$ (defined analogously to Definition~\ref{def:itcomp} on Page~\pageref{def:itcomp})
 only includes a finite number of MPs. On the other
hand we will prove that it gives a sound upper bound for a sequence of transitions.

\begin{lem}%
\label{lem:taucomp-ub}
Let $\vec p$, $\vec q$ be $\tau$-closed, similar IMPs and let $\sigma$, $\rho$ be weighted traces such that
$\vec p$ bounds $\sigma$ and
   $\vec q$ bounds $\rho$. Then
   ${\vec q}\taucomp {\vec p}$ bounds $\sigma\rho$.
\end{lem}
\noindent
Compared to Lemma~\ref{lem:comp-ub}, this lemma poses stronger conditions, but has a stronger conclusion since
${\vec q}\taucomp {\vec p}$ is, in general, lower than ${\vec q}\circ {\vec p}$.

\begin{proof}
The definition of ${\vec q}\taucomp {\vec p}$ makes the statement trivial for variables that are not self-dependent.
So, let $x_i$ be self-dependent.  Suppose that $\tseq{\vec x}{\sigma}{\vec y}\tseq{\:}{\rho}{\vec z}$.
Let $t_s = \wgt{\sigma}$, $t_r = \wgt{\rho}$.
By assumption,
\begin{align*}
 z_i & \le  {\vec q}(\vec y, t_r) & \\ &
  = y_i +  t_r{\vec q}[i]'(\vec y, t_r) \\  &
\le \left( x_i +  t_s{\vec p}[i]'(\vec x, t_s) \right) +  t_r{\vec q}[i]'(\vec y, t_r) \\  &
\le \left( x_i +  t_s{\vec p}[i]'(\vec x, t_s+t_r) \right) +  t_r({\vec q}[i]' \circ \vec p)(\vec x, t_s+t_r) \\  &
\le x_i +  (t_s+t_r)  ({\vec p}[i]' \join ({\vec q}[i]' \circ \vec p))(\vec x, t_s+t_r) \tag{by Lemma~\ref{lem:afplusbg}} \\
&  =  ({\vec q}\taucomp {\vec p})[i](t_s+t_r) \,. \tag*{\qedhere}
\end{align*}
\end{proof}

\begin{lem}[The Finite Closure Lemma]%
\label{lem:tau-finiteness}
Let $\setT$ be a finite set of $\tau$-closed, similar
 IMPs.  Then the set
\[
 \setT^{\bigstar} = \bigcup_{k=1}^\infty \setT^{\taucomp (k)}
 \]
 is finite.
Hence,  there is a number $\ell$ such that $\setT^{\bigstar} = \setT^{\taucomp (\ell)}$.
\end{lem}

\begin{exa}
Consider $\setT = \{{\vec p}, {\vec q}\}$ with $\vec p$ from Example~\ref{ex:tauproduct} and
\[
\begin{array}{lrlllll}
& {\vec q} = \langle & x_1 + \tau x_4, & x_2, & x_3 + \tau x_4,  & x_4 \rangle .
\end{array}
\]
Then
\[
\begin{array}{lrlllll}
& {(\setT)}^{\taucomp (2)} = & \setT \taucomp \setT  = & \{  {\vec p}\taucomp {\vec p},  & {\vec q}\taucomp {\vec p},   &{\vec q}\taucomp {\vec q}, &   {\vec p}\taucomp {\vec q} \} = \\
&&& \{ {\vec p}, & {\vec r}, & {\vec q}, & {\vec r} \}
\end{array}
\]
Where
\[
\begin{array}{lrlllll}
& {\vec r} = \langle & x_1 + \tau x_4, & x_2 + \tau x_1, & x_3 + \tau x_2 + \tau^2 x_1 + \tau x_4,  & x_4 \rangle .
\end{array}
\]
Next, ${(\setT)}^{\taucomp (3)}$ consists of 8 products, which contribute only one new element to the accumulated union, namely
\[
\begin{array}{lrlllll}
& {\vec p}\taucomp {\vec r} = \langle & x_1 + \tau x_4, & x_2 + \tau x_1, & x_3 + \tau x_2 + 2\tau^2 x_1 +  \tau x_1 +\tau x_4,  & x_4 \rangle .
\end{array}
\]
Adding  ${(\setT)}^{\taucomp (4)}$ to the union, we only have to check ${\vec p}\taucomp {\vec p}\taucomp {\vec r}$ and ${\vec q}\taucomp {\vec p}\taucomp {\vec r}$.
It turns out that these $\tau$-MPs have appeared already.  Therefore, we have reached $\setT^{\bigstar} $.
\end{exa}

\begin{proof}
Define the following sequence:
\begin{align*}
T_{(0)} & = \idppol \\
T_{(k+1)} & =  T_{(k)}  \join \    ( \bigsqcup \setT ) \circ T_{(k)} \, .    
\end{align*}

We prove the following statement by complete induction on $k$:
\begin{equation} \label{eq:Tstar-sdv}
  {\vec r} \in \setT^{\bigstar},\  i\le k  \ \Rightarrow {\vec r}[i] \lepoly  T_{(k)} [i] \,.
\end{equation}
This implies the desired result because it shows that for any $i\le n$, the set of values assumed by ${\vec r}[i]$, when $\vec r$ ranges
over $\setT^{\bigstar}$, is finite, because it is contained in the down-set of a single polynomial, namely $T_{(i)}[i]$,
and such a down-set is always finite. Hence, $\setT^{\bigstar}$ is finite.
It is worthwhile to observe that since the members of $\setT$ are similar IMPs,  then,
whenever $i$ is a self-dependent variable (in any of them---and then in all), $\bigsqcup \setT$ also has the form $x_i + \text{(function of lower variables)}$
and, consequently, so do the $T_{(k)}$.

To prove~\eqref{eq:Tstar-sdv},
fix $k$  and $i\le k$.  We now employ induction on the number $t$ of compositions used to construct $\vec r$ (following the definition of $\setT^{\bigstar}$).
The base case is $t=1$, i.e., ${\vec r} \in \setT$ where the result is immediate.
In the general case,
\[
 {\vec r} = {\vec q}\taucomp {\vec f}
 \]
 where ${\vec q} \in \setT$ and ${\vec f}\in \setT^{\taucomp (t-1)}$.

\emph{Case~}1: we consider ${\vec r}[i]$ when $i\in \sd{\vec q}$.
Write
 \[
  {\vec q}[i] = x_i + \tau {\vec q}[i]'
 \]
 then
 \[
  {\vec r}[i] = (x_i + \tau {\vec f}[i]') \join \tau ( {\vec q}[i]' \circ {\vec f} )
 \]
 Note that ${\vec q}[i]'$ only depends on variables $x_j$ with $j < i \le  k$.
By the induction hypothesis (of the induction on $k$) we have
\[
 {\vec f}[j] \lepoly T_{(k-1)} [j]
\]
Consequently
 \[
  \tau ({\vec q}[i]' \circ {\vec f}) 
      \lepoly \left({\vec q}\circ T_{(k-1)}\right) [i]
       \lepoly T_{(k)} [i] \, .
 \]
By the induction hypothesis on $t$,
\[
 {\vec f}[i] = x_i + \tau {\vec f}[i]' \lepoly T_{(k)} [i]
\]
and we conclude that
\[
    (x_i + \tau {\vec f}[i]') \join \tau ( {\vec q}[i]' \circ {\vec f} )   \lepoly
       T_{(k)} [i] \, \join \, T_{(k)} [i] =
       T_{(k)} [i] \, .
 \]

\emph{Case~}2: consider $i\notin \sd{\vec q}$.
Then ${\vec q}[i]$ only depends on self-dependent variables lower than $i$, hence lower than $k$.
By the induction hypothesis on $k$, when $j<k$ we have
\[
 {\vec f}[j] \lepoly T_{(k-1)} [j]
\]
Consequently
 \[
   {\vec r}[i] = {\vec q}[i] \circ {\vec f}  \lepoly  T_{(k)} [i] \, . \qedhere
 \]
\end{proof}

We are heading towards a central part of the proof, where we construct upper bounds for sequences of MPs taken from a set $\setS$ of
 similar $\tau$-MPs. Moreover, we restrict attention to MPs $\vec p$ such that $\tra{\vec p}$ is idempotent.
Lemma~\ref{lem:tau-finiteness} is used to show that a finite set of upper bounds covers such traces of any length. However, it requires the MPs to be
$\tau$-closed IMPs.  We close this gap by demonstrating that we can obtain upper bounds for such traces using  $\sdm{\setS}^\tau$
instead of $\setS$ (the $\tau$ superscript over a set name means that we apply generalization to all members of the set).


\begin{lem}%
\label{lem:bounding-by-sdm}
Let $\setS$ be a set of  similar $\tau$-MPs such that $\tra{\vec p}$, for ${\vec p}\in \setS$, is neat.
For any $t>n$, 
consider a concrete trace
\[ {\concrete{\sigma}} = \wtrans{{\vec p}_1}{s_{0}}{s_1}{w_1}\dots\wtrans{{\vec p}_t}{s_{t-1}}{s_t}{w_t} \,.\]
Let ${\vec r} = \sdm{{\vec p}_{t}\:\dots\:{\vec p}_{t-n+1}}$.
Then $s_t \le {\vec r}(s_{t-1}, w')$ where $w' = \sum_{z=t-n+1}^t w_z$.   In other words, ${\vec r}$ bounds the last step of $\sigma$ provided that its weight is
suitably increased.
\end{lem}

The point is to replace ${\vec p}_t$ by $\vec r$ as a bound on the last step, the gain being the convenient form of $\vec r$, namely an IMP\@.

\begin{proof}
We use induction on $d$ to prove the following claim for all $0 < d \le n$:

\bigskip\noindent
Let  ${\vec r}_d = \sdm{{\vec p}_{t-n+d}\:\dots\:{\vec p}_{t-n+1}}$ and let $w'_{d} = \sum_{z=t-n+1}^{t-n+d} w_z$.

\bigskip\noindent
Then for $i\le d$, $s_{t-n+d}[i] \le {\vec r}_d (s_{t-n+d-1},w'_{d})$.

\bigskip\noindent
Note that this claim implies the lemma's statement (setting $d=n$). We now move to its proof.

\bigskip\noindent
\emph{Base case}: $d = 1$. Hence ${\vec r}_d = {\vec p}_{t-n+1}$. The claim only concerns $i=1$, and
$s_{t-n+1}[i] \le {\vec p}_{t-n} (s_{t-n+1},w'_{1})[i]$ by definition.

\bigskip\noindent
\emph{Inductive case}: $d > 1$.  Then
\[
 {\vec r}_d  =  {\vec p}_{t-n+d} \circ \sdpart{{\vec p}_{t-n+d-1}} \circ \dots \circ \sdpart{{\vec p}_{t-n+1}} \,.
\]
Consider $s_{t-n+d}[i]$, for some  $i \le d$.  By assumption,
\begin{equation} \label{eq:st-l}
 s_{t-n+d}[i] = {\vec p}_{t-n+d} [i] (s_{t-n+d-1},w_{t-n+d}) \,.
\end{equation}
Let us consider the entries $s_{t-n+d-1}[j]$ on which the above expression may depend;
thus  $j\le i$.
Let us focus first on the case that $x_j$ is self-dependent. We observe that
\begin{equation} \label{eq:j-sd}
s_{t-n+d-1}[j]  =   ( \sdpart{{\vec p}_{t-n+d-1}} \circ \dots \circ \sdpart{{\vec p}_{t-n+1}} ) [j] (s_{t-n+d-1},w)
\end{equation}
for any $w$, since $ \sdpart{\vec p}[j] = x_j$ for all ${\vec p}\in \setS$
 (note that both sides of the equation refer to $s_{t-n+d-1}$).

When $x_j$ is not self-dependent, we have (using the induction hypothesis, and the definition of self-dependent cut)
\begin{equation} \label{eq:j-not-sd}
\begin{aligned}
s_{t-n+d-1}[j]  &\le   ( {\vec p}_{t-n+d-1} \circ \sdpart{{\vec p}_{t-n+d-2}} \circ \dots \circ \sdpart{{\vec p}_{t-n+1}} ) [j] (s_{t-n+d-2},w'_{d-1})
                 \\        &=   ( \sdpart{{\vec p}_{t-n+d-1}} \circ \dots \circ \sdpart{{\vec p}_{t-n+1}} ) [j] (s_{t-n+d-2},w'_{d-1})
\end{aligned}
\end{equation}
We now wish to replace $s_{t-n+d-2}$ in the last expression by $s_{t-n+d-1}$.  To this end
we apply Lemma~\ref{lem:sdm-property}, establishing that the entries $s_{t-n+d-2}[e]$ that influence the last expression are either self-dependent
or have index $e\le j-d+1 \le 1$. But in the latter case $e=1$ and $x_1$ is certainly self-dependent. We conclude that
${\vec p}_{t-n+d-1}[e] \ge x_ e$, hence $s_{t-n+d-1}[e] \ge s_{t-n+d-2}[e]$.  We now obtain from~\eqref{eq:j-not-sd} that
\begin{align*}
s_{t-n+d-1}[j]  &\le     ( \sdpart{{\vec p}_{t-n+d-1}} \circ \dots \circ \sdpart{{\vec p}_{t-n+1}} ) [j] (s_{t-n+d-2},w'_{d-1})  \\
                        &\le ( \sdpart{{\vec p}_{t-n+d-1}} \circ \dots \circ \sdpart{{\vec p}_{t-n+1}} ) [j] (s_{t-n+d-1},w'_{d-1}) \,.
\end{align*}

\noindent
We substitute this in~\eqref{eq:st-l} to obtain
\begin{align*}
 s_{t-n+d}[i] &\le  ( {\vec p}_{t-n+d} [i] \circ  ( \sdpart{{\vec p}_{t-n+d-1}} \circ \dots \circ \sdpart{{\vec p}_{t-n+1}} ) )
(s_{t-n+d-1},  w'_{d-1} + w_{t-n+d}) \\
 &=  {\vec r}_d (s_{t-n+d-1},w'_{d})  \,. \qedhere
\end{align*}
\end{proof}


\begin{lem}%
\label{lem:idempotent-ub}
Let $\setS$ be a set of similar $\tau$-MPs such that $\tra{\vec p}$, for all ${\vec p}\in \setS$, is neat.
 Then every weighted trace over $\setS$ of length greater than $n$ has an upper bound in the set
\( \left(\left.(\sdc{\mathcal S}^\tau)\right.^{\bigstar}   \circ {\mathcal S}^{(n)}\right) [n\tau / \tau]  \).
\end{lem}

\begin{proof}
Write such a trace as  $\sigma = \wtrans{{\vec p}_1}{s_{0}}{s_1}{w_1}\dots\wtrans{{\vec p}_t}{s_{t-1}}{s_t}{w_t}$.
Let $\sigma_i$ be the $i$th step, namely $\wtrans{{\vec p}_i}{s_{i-1}}{s_i}{w_i}$.
For each $i > n $, by the last lemma,  the $i$th transition is bounded by a $\tau$-MP  ${\vec r}\in  \sdc{\mathcal S}$,
with modified weight.  If $\vec r$ bounds a certain transition, then ${\vec r}^\tau$ certainly does.
Consequently,
by Lemmas~\ref{lem:comp-ub} and~\ref{lem:taucomp-ub},
$\sigma_1 \dots \sigma_n  \sigma_{n +1}' \dots \sigma_{t}'$  (where the primed transitions use modified weights)
is bounded by a $\tau$-MP ${\vec q} \in \left.(\sdc{\mathcal S}^\tau)\right.^{\bigstar} \circ {\mathcal S}^n$.
In order to get rid of the modified weight, we note that
\begin{align*}
 \wgt{\sigma_1 \dots \sigma_n   \sigma_{n +1}'\dots \sigma_t'}  & = \wgt{\sigma_1 \dots \sigma_n  } + \sum_{i=n +1}^t \wgt{\sigma_{i}'}  \\
  &  =  \sum_{i=1}^n  w_i +  \sum_{i=n +1}^t \,\sum_{z=i-n+1}^{i} w_z \\
   & \le n \wgt{\sigma}
 \end{align*}
 Thus we see that ${\vec q}[n \tau/\tau]$ bounds $\sigma$ using its original weight.
 \end{proof}

Let $\setS$ be the SDL under analysis, and $\alpha(\setS)$ map it into $\absppol$. The latter is a monoid with respect to the composition operation,
and we obtain a monoid homomorphism $\alpha: \tseqs{\setS} \to \absppol$ from abstract traces to $\absppol$.
The \emph{Factorization Forest Theorem}
of Imre Simon~\cite{Simon:TCS:90} shows that such a homomorphism induces a useful structure on the traces, provided that the codomain is a
\emph{finite} monoid.  Now, $\absppol$ is, of course, infinite; but if we assume that the loop under analysis is polynomially bounded, then,
as argued in the proof of termination (Corollary~\ref{cor:termination}), the closure $\closure{\alpha(\setS)}$ is finite. Note that it is a sub-monoid of $\absppol$.
Thus, Simon's theorem can be applied. Instead of the original formulation,
we can use a convenient corollary of Simon's theorem, stated by Boja\'nczyk~\cite{Bojanczyk:DLT:09} (actually we will not use his formulation but
a simplified one, as we do not need its full power).

In the statement of this result, we refer to $\alpha(\sigma)$, with $\sigma\in \tseqs{\setS}$,  as the \emph{type} of $\sigma$.
We consider regular expressions constructed using the operators: concatenation, union and Kleene-plus (where $E^+$ generates the union of all
languages generated by $E^i$ for $i>0$).
A regular expression $E$ over the alphabet $\setS$ is \emph{well-typed} if for each of its sub-expressions $F$ (including $E$),
all words (traces) generated by $F$ have the same type, which is then the type of the expression.

\begin{thm}[Boja\'nczyk]%
\label{thm:regexp}
The existence of a homomorphism $\alpha: {\setS}^* \to M$, where $M$ is a finite monoid, implies that $\setS^*$ can be generated by
a finite union of well-typed regular expressions.
\end{thm}

Now, all we have to do is prove that the set of AMPs returned by our algorithm provides upper bounds for all the traces generated by each of these
regular expressions.  Note that this theorem highlights the role of \emph{idempotence} in $\absppol$, since for an expression $E^+$ to be well-typed,
$\alpha(\sigma)$ has to be the same for all words $\sigma$ generated by $E^+$, which implies that it is an idempotent element.

\begin{thm}%
\label{thm:SDL-ub}
Let $\setS$ be a polynomially-bounded SDL\@.  Let $\setA$ be the set of abstract $\tau$-MPs returned by Algorithm \procSDL{}.  Then
there is a finite set $\setB \subset \tppol$ such that $\alpha(\setB) \subseteq \setA$
and for any trace $\sigma\in \tseqs{\setS}$
there exists
${\vec p}\in\setB$ that bounds $\sigma$.
\end{thm}

\begin{proof}
We consider the regular expressions established by Theorem~\ref{thm:regexp} and all their sub-expressions: this is a finite set.
We construct a set of bounds ${\setB}_E$, with $\alpha({\setB}_E)\subseteq \setA$,
 for each such sub-expression $E$, by structural induction on the expressions. Clearly, this proves the theorem.
Importantly, the construction maintains these properties:
\begin{itemize}
\item If expression $E$ has type ${\vec a}\in \absppol$ then every ${\vec p}\in {\setB}_E$ has
$\tra{\vec p} = {\vec a}$
 (intuitively, the algorithm does not change the shape of the multi-polynomials except by adding $\tau$'s).
\item Moreover, $\vec p$ is realizable (this is not hard but requires a bit of attention since Theorem~\ref{thm:SDL-lb} does not refer to the
 bounds constructed in the current proof, but only to $\setA$).
\end{itemize}
If $E$ is a single transition $\vec p$ we set $\setB = \{\vec p\}$.

\bigskip\noindent
If $E$ is $FG$ we compose ${\setB}_F$ with ${\setB}_G$. Note that $\alpha({\setB}_G \circ{\setB}_F) = \alpha({\setB}_G) \acirc \alpha({\setB}_F) \subseteq \setA$
 thanks to the closure computation.  Realizability follows from Lemma~\ref{lem:comp-realizable}.

\bigskip\noindent
If $E$ is  $F+G$ we unite ${\setB}_F$ with ${\setB}_G$. Again, the abstraction of the result is in $\setA$.

\bigskip\noindent
It remains to consider an expression of the form $F^+$.
Consider ${\setB}_F$. By IH, it consists of realizable $\tau$-MPs and are all similar to a single idempotent AMP,
so the requirements of Lemma~\ref{lem:idempotent-ub} are satisfied. Let $\Phi$ be the set of traces generated by $F$. Then every $\sigma\in\Phi$
has a bound in ${\setB}_F$.  For a concatenation of $i\le n$ such traces we have a bound in ${({\setB}_F)}^{(i)}$. For a concatenation of
more than $n$ traces, consider each of these traces as a weighted transition where the weight represents the trace's length.
By Lemma~\ref{lem:idempotent-ub}, the concatenation of the traces has a bound in
\begin{equation}
    \left(\left.(\sdc{\setB_F}^\tau)\right.^{\bigstar}   \circ {\setB_F}^{(n)}\right) [n\tau / \tau] \,. \label{eq:bounds}
\end{equation}
By Lemma~\ref{lem:tau-finiteness},
we can replace $\left.(\sdc{\setB_F}^\tau)\right.^{\bigstar}$ with
$\left.(\sdc{\setB_F}^\tau)\right.^{\taucomp (\ell)}$ for some $\ell > 0$.
Moreover, it is sound to relax the upper bound to $\left.({\setB_F}^\tau)\right.^{\taucomp (\ell)}$. Now~\eqref{eq:bounds} becomes
\[
    \left(\left.({\setB_F}^\tau)\right.^{\taucomp (\ell)}   \circ {\setB_F}^{(n)}\right) [n\tau / \tau] \,.
\]
We claim that this set of upper bounds satisfies all our requirements. For realizability, all we need is the observation that substituting
$n\tau$ for $\tau$ does not affect realizability; plus Lemma~\ref{lem:comp-realizable} and the fact
${\vec q}\taucomp {\vec p} \le {\vec q}\circ {\vec p}$.

Finally we look at the abstractions of these $\tau$-MPs, namely the set
\[
    \alpha\left(\left(\left.({\setB_F}^\tau)\right.^{\taucomp (\ell)}   \circ {\setB_F}^{(n)}\right) [n\tau / \tau]\right) =
   \left.({\alpha({\setB_F})}^\tau)\right.^{\acirc (\ell)}  \acirc {(\alpha({\setB_F}))}^{\acirc (n)}
\,.
\]
These AMPs are included
in the result of \procSDL{}, since they they are produced by closure, generalization and closure again.
Note also that the bounds conform with the type of $F$ (when $\tau$'s are ignored).
\end{proof}

\newcommand{\cscc}[0]{\boxplus}

\section{On the Computational Complexity of our Problem}%
\label{sec-complexity}

Our main goal in this research was to establish that the problem of computing tight bounds is solvable.  However, once proved solvable,
the question of its complexity arises.
For simplicity we assume that we are only dealing with programs where all variables are polynomially bounded.
We consider the complexity in terms of three parameters: $|P|$, the size of the program; $n$, the number of variables; and $d$, the highest degree reached.
We note that if the user wishes to verify a desired degree bound $d$, say check that a program is at most of cubic time complexity,
it is possible to use an abstraction that truncates exponents higher than $d$ and the complexity will be reduced accordingly.
First, we give an upper bound.%
\footnote{To avoid any confusion, we are obliged to point out that in the preliminary version (proceedings of FoSSaCS 2019) a wrong expression for the
upper bound was given.}

\begin{thm}
Our algorithm runs in time polynomial in $|P|\cdot {2^{n^{d+1}}}$.
\end{thm}

\begin{proof}
We estimate the complexity of our algorithm, based on bounding the number of different AMPs that may be encountered.
The number of monomials over $n$ variables is bounded (for $n>2$) by $n^d$ (think of a monomial as a product of at most $d$ variables.
For uniqueness list the indices in descending order. The number of descending lists of length at most $d$ is bounded by $n^d$).
 Since a MP is an $n$-tuple of sets of monomials, the number of possible AMPs
is less than ${(2^{n^d})}^n = 2^{n^{d+1}}$.

Next, we consider the running time of the analysis of a loop (independent of the position of the loop
in the program, so we can later just multiply this time bound by the number of loops in the program).
 In  procedure \procSDL{}, a set of AMPs is maintained and repeatedly enlarged (by applying closure and generalization), until stable.
The time to perform each round of enlargement is polynomial in the size of the resulting set (more precisely its representation, but with any reasonable implementation
this does not change much) and the number of rounds is clearly bounded by the size of the final set.   So we deduce that the total time is polynomial in
$2^{n^{d+1}}$.

Finally we should add the time to  represent non-looping code as a set of AMPs, and other ``book-keeping'' operations, but clearly they contribute at most a polynomial in $|P|$ and the number of AMPs.
\end{proof}

Is this a satisfactory upper bound?  It seems high, and is probably not tight. We know, however, that a solution to our problem must use at least exponential time
in the worst case,
because it has a potentially high \emph{output size}.
\begin{clm}
There is a  command, of size polynomial in $n$ and $d$, which requires \[{(\lfloor n/(2d)\rfloor)}^{\lfloor nd/2\rfloor }\] AMPs  to describe its result.
\end{clm}

\begin{proof}
We assume that $n$ is even and a multiple of $d$, so we can avoid the ``floor'' signs.
We use $m = n/2$ variables called $\X_1,\dots,\X_{m}$ and $m$ variables called $\Y_1,\dots,\Y_{m}$.
As usual let $x_i$ denote the initial value in $\X_i$.
For each $j\le m$ we write a piece of code that computes
\[  (x_1 \lor x_2 \lor \dots \lor x_{m/d}) \ast \cdots \ast (x_{m-(m/d)+1} \lor  \dots \lor x_{m}), \]
where the disjunction operator represents a non-deterministic choice. Hence we choose $d$ values and multiply them together.
It should be clear that this produces one of ${(m/d)}^d$ uncomparable monomials, and the result of analysing this command will have to
list all of them.  We assign the product to $\Y_j$.  Since the non-deterministic
choices are made independently for each $j$, to describe the outcome in $\Y_1,\dots,\Y_{m}$ we need ${({(m/d)}^d)}^m = {(m/d)}^{md}$ uncomparable AMPs.
\end{proof}

Thus, as long as we use an explicit AMP representation for the output of our algorithm, worst-case
complexity exponential in $nd$ is unavoidable---at least for $d$ smaller than $n$.
This does not necessarily rule out the applicability of the algorithm, as some algorithms susceptible to combinatorial explosion still prove usable in static analysis
applications (consider the closure-based algorithm
for Size-Change Termination).  At any rate, we intend to complement the work presented in this article by further research into improving the algorithm,
which we consider only a starting point, being the first complete algorithm for this problem.  We now list some speculations about this future work.
The reader may have noticed that for the ``bad'' example we used above, the description would be very compact if we were allowed to use the $\textbf{max}$ 
operator and embed it in expressions (for our output bound for each $\Y_j$ is just a product of $d$ ``max'' expressions). But this does not seem to be a panacea,
and we suspect that exponential output size---and running time, for sure---may be necessary even with \textbf{max}. We leave this as an open problem;
of course, the central problem is to identify the complexity class of the analysis problem.
We can also ask what happens if we redefine our problem so that the output size is small---a natural example is the case of a univariate bound;
or ask for an output-size dependent complexity function. In these cases
the output-size excuse for exponential complexity does not apply.
Another open problem that is raised by the above considerations is: what is the best bound on $d$ in terms of $|P|$ and $n$ alone?

\section{Algorithm Extensions and Open Problems}%
\label{sec:extensions}

In this section we list some ideas about how this research might be extended, specifically in terms of adding features to the subject language
(while keeping the completeness of the algorithm!), and whether we believe that our current approach
suffices for solving these extensions.

\subsection{``Unknown'' value}

When abstracting real-world code into a restricted language, it is common to have cases in which a value
has to be treated as ``unknown.'' It might be really unknown (determined by the environment) or the result of computations
that we cannot model in the restricted language (note that in our case, if we can \emph{bound} a computation by a polynomial expression we
are happy enough). It seems useful to extend our language by a special ``unknown'' value.  Another example of its usage,
following~\cite{JK08}, is to analyze the growth of variables in loops for which no iteration bound is known;
in this situation one cannot obtain a time bound for the program, but we can
still  bound computed values and may be able to draw conclusions regarding other quantities of interest, perhaps space complexity.
We simulate such a loop using the ``unknown'' value as the loop bound.
Concretely, this extension can be implemented by using a dedicated variable $x_u$ for anything unknown, and, throughout the algorithm,
replacing any expression that includes $x_u$ immediately by $x_u$. This prevents an explosion of the MP set (or even failure to reach a fixed point) because
of expressions including $x_u$.

\subsection{Resets}

\cite{B2010:DICE} extended the decidability result from~\cite{BJK08} to a language that contains the \emph{reset} command 
\verb/X := 0/.  This addition may seem trivial at first, however for the problem of deciding polynomial boundedness it caused the
complexity of decision to jump from PTIME to PSPACE-complete. The increased complexity arises from the need to recognize the situation
that a variable ``is definitely zero,'' subsequent to a particular execution path. In this case, a loop that has this variable as counter will
\emph{definitely} not execute. So the algorithm has to deal with tracking these 0's around. However, our algorithm
already tracks  data-flow rather precisely and the algorithm is exponential anyway. We believe that
our algorithm can be extended to handle
resets without further raising the complexity of the solution.

\subsection{Flowchart programs}

In~\cite{BAPineles:2016} the results of~\cite{BJK08} (and, implicitly, also~\cite{B2010:DICE}) are extended from a structured language,
that can be analyzed in a compositional manner (as we have done), to a ``flowchart'' language, where a program is presented as a
\emph{control-flow graph}, or flowchart, of arbitrary shape, together with \emph{annotations} that convey information regarding
iteration bounds. We argued there that this program form is more general and closer to the form used by several analysis tools for
real-world programs. The results of~\cite{BJK08} were carried over to this language by transforming flowchart programs into programs
in a well-structured language LARE that is  slightly more expressive than our core language. It seems that the same development should be doable
with precise polynomial-bound analysis, we only have to extend our algorithm to the language LARE\@. We have not investigated this in detail yet.

\subsection{Deterministic loops}

In~\cite{BK11}, Kristiansen and Ben-Amram looked at a variant of our core language where loops are deterministic: the semantics of
\verb/loop X {C}/ is to perform \pgt{C} \emph{precisely} as many times as the value of \pgt{X}.  It was shown that the decision problem
of polynomial boundedness, addressed in~\cite{BJK08}, becomes undecidable in this case; however the undecidability proof exploits
worst-case scenarios where some variables are constant while others grow.  It is conjectured that the problem is decidable if one only
asks about bounds that are either univariate, or multivariate but asymptotic in all variables. We propose the same conjecture with respect
to the problem of tight polynomial bounds.

\subsection{Increments and Decrements}

In our opinion, the feature that most strikingly marks our language as weak is the absence of increments and decrements (and explicit constants
in general; but if you have increment and reset you can generate other constants). However, anyone familiar with counter machines will realize that
including them would bring
 our language very close to counter machines and hence to undecidability; still, without deterministic loops, this model falls just \emph{a little}
short of counter machines. We pose the decidability of the polynomial bound problem in such a language (and the investigation of the model from a
computability viewpoint, in general) as a challenging open problem.

\subsection{Procedures}

Due to the compositional form of our algorithm, we suppose that extension of the language with first-order, non-recursive procedures is not hard, but have not
investigated this further.  A much greater challenge is to allow recursive procedures (one has to figure out what is a good way to do that since, of course,
we cannot allow unbounded recursion). Another one is to include high-order functions.
Avery, Kristiansen and Moyen outline in~\cite{AKM09} a possible approach
for promoting analyses like~\cite{BJK08} to higher-order programs, but a definite decidability result is not obtained.

\section{Related Work}%
\label{sec:rw}

Bound analysis, in the sense of finding symbolic bounds for data values, iteration bounds and related quantities,
is a classic field of program analysis~\cite{Wegbreit:75,Rosendahl89,ACE}. It is also an area of active research, with tools being
currently (or recently) developed including \tool{COSTA}~\cite{Albert-et-al:TCS:2011}, \tool{AProVE}~\cite{APROVE-JAR2017},
\tool{CiaoPP}~\cite{CiaoPP-TPLP2018} 
, $C^4B$~\cite{CHS:pldi2015},  
 \tool{Loopus}~\cite{SZV:jar2017}---and this is just a sample of tools for imperative programs. There is also work on functional and logic programs,
term rewriting systems, recurrence relations, etc.~that we cannot attempt to survey here.  In the rest of this section we point out work that is more directly
related to ours, and has even inspired it.

The LOOP language is due to Meyer and Ritchie~\cite{MR:67}, who note that it computes only primitive recursive functions, but complexity can rise very fast, even for programs with nesting-depth 2.
Subsequent work concerning similar languages~\cite{KasaiAdachi:80,
KN04,NW06,JK08}
attempted to analyze such programs more precisely;
most of them proposed syntactic criteria, or analysis algorithms, that are
sufficient for ensuring that the program lies in a desired class (often, polynomial-time programs),
but are not both necessary and sufficient: thus, they do not prove
decidability (the exception is~\cite{KN04} which has a decidability result for a weak ``core'' language).
The core language we use in this paper is from~Ben-Amram et al.~\cite{BJK08}, who observed that by introducing weak bounded loops instead of concrete loop commands and
non-deterministic branching instead of ``\texttt{if},'' we have weakened the semantics just enough to obtain decidability of polynomial growth-rate.
This research was motivated by observing that all the previous algorithms, although they implicitly relax the semantics (since they do not analyze
conditionals, etc.), do not provide completeness over the core language.
 Justifying the necessity of these relaxations,~\cite{BK11} showed undecidability for a language that can only do addition and definite loops (that cannot exit early).

In the vast literature on bound analysis in various forms,
 there are a few other works that give a complete solution for a weak language.
 \emph{Size-change programs} are considered by~\cite{CDZ-MFCS14,Zuleger-CSR15}. Size-change programs abstract away nearly everything in the program,
leaving a control-flow graph annotated with assertions about variables that decrease (or do not increase) in a transition. Thus, it does not assume structured
and explicit loops, and it cannot express information about values that increase. Both works yield tight bounds on the number of transitions until termination.

Dealing with a somewhat different problem,~\cite{Seidl-polynomial-invariants,Ouaknine-polynomial-invariants}
 both check, or find, \emph{invariants} in the form of polynomial equations. We find it remarkable that they give
complete solutions for  weak languages, where the weakness lies in the non-deterministic control-flow, as in our language.
  If one could give a complete solution for polynomial
\emph{inequalities}, this would imply a solution to our problem as well.

The \emph{joint spectral radius} problem for semigroups of matrices is related to our work as well (though we have not discovered this until recently,
due to the work being done in an entirely different context). Specifically,~\cite{JPB:2008} gives an algorithm that can be expressed in the language of our
work as follows: given a Simple Disjunctive Loop in which all polynomials are linear, the algorithm decides if the loop is polynomially bounded and if it is,
returns the highest degree of $\tau$ in the tight polynomial upper bound (over all variables). A closer inspection shows that it can actually determine the degree
in which $\tau$ enters the bound for every variable. Thus, it solves a certain aspect of the SDL analysis problem. Their algorithm is polynomial-time and uses an
approach similar to~\cite{BJK08}.

\section{Conclusion and Further Work}%
\label{sec:conclusion}

We have solved an open problem in the area of analyzing programs in a simple language with bounded loops. For our language, it has been previously
shown that it is possible to decide whether a variable's value, number of steps in the program, etc.,~are polynomially bounded.
Now, we have an algorithm that computes tight polynomial bounds on the final values of variables in terms of initial values.
The bounds are tight up to constant factors (suitable constants are also computable).
This result improves our understanding of what is computable by, and about, programs of this
form.  An interesting corollary of our algorithm is that as long as variables are \emph{polynomially bounded}, their worst-case bounds are
described tightly by (multivariate) \emph{polynomials}.  This is, of course, not true for  common Turing-complete languages.
Another interesting corollary of the \emph{proofs} is the definition of a simple class of patterns that suffice to realize the worst-case behaviors.

There are a number of possible directions for further work.
\begin{itemize}
\item
As discussed in Section~\ref{sec-complexity}, we have not settled  the computational complexity of the problem we have solved.
\item
We propose to look for decidability results for richer (yet, obviously, sub-recursive) languages. Some possible language extensions
include deterministic loops, variable resets (cf.~\cite{B2010:DICE}), explicit constants, and (recursive) procedures.
The inclusion of explicit constants is a particularly challenging open problem.
\item
Rather than extending the language, we could extend the range of bounds that we can compute.
In light of the results in~\cite{KN04}, it seems plausible that the approach can be extended to classify the
Grzegorczyk-degree of the growth rate of variables when they are super-polynomial. There may also be
room for progress regarding precise bounds of the form $2^{poly}$.
\item
Our algorithm computes bounds on the highest possible values of variables. With our restricted arithmetic we can reduce the calculation of
a ``countable resource'' like the number of steps to bounding a variable.  However,
our weak language seems useless as an abstraction for more advanced resource-analysis problems,
e.g., analysis of expected costs (for programs in which such an analysis is interesting).
So we pose the problem of designing weak languages that adequately abstract some non-trivial cases of more advanced analyses and obtain computability
for them.
\item
Finally, we hope to see the inclusion of our algorithm (or at least the approach) in a system that handles a real-life
programming language. In particular, it would be interesting to see how our method works together with techniques that discover loop bounds,
typically via ranking functions.
\end{itemize}

\bigskip
\paragraph{\bf Acknowledgment.} Amir M. Ben-Amram is grateful for the hospitality at the School of Computing, Dublin City University, where part of this work has been done.  The authors also thank the referees for valuable comments.

\bibliographystyle{alpha}
\newcommand{\etalchar}[1]{$^{#1}$}

\end{document}